\documentclass[11pt,a4paper]{article}
\pdfoutput=1

\usepackage[utf8]{inputenc}
\usepackage[english]{babel}
\usepackage{chicago}

\usepackage[round]{natbib}
\usepackage[margin=1in]{geometry}

\usepackage{graphicx}
\usepackage{amsmath}
\usepackage{amssymb}
\usepackage{amsthm}
\newtheoremstyle{prop}
  {3pt}
  {3pt}
  {}
  {}
  {\scshape}
  {.}
  {.5em}
  {}
\theoremstyle{prop}
\newtheorem*{proposition}{Proposition}
\newtheoremstyle{def}
  {3pt}
  {3pt}
  {}
  {}
  {\scshape}
  {.}
  {.5em}
  {}
\theoremstyle{def}
\newtheorem*{definition}{Definition}

\usepackage{amsfonts, multirow, booktabs, url, verbatim}
\usepackage[english]{babel}
\usepackage[utf8]{inputenc}
\usepackage{color}
\definecolor{green}{rgb}{0,0.75,0}

\newcommand{\X}{\boldsymbol{X}}
\newcommand{\x}{\boldsymbol{x}}
\newcommand{\Xobs}{\mathbf{X}}
\newcommand{\Z}{\boldsymbol{Z}}
\newcommand{\Zest}{\hat{\Z}}
\newcommand{\M}{\boldsymbol{M}}

\newcommand{\Ysl}{\tilde{Y}}
\newcommand{\T}{{}^{\top}}

\newcommand{\rank}{\mathrm{rank}}
\newcommand{\mub}{\boldsymbol{\mu}}
\newcommand{\betab}{\boldsymbol{\beta}}
\newcommand{\Sigmab}{\boldsymbol{\Sigma}}
\newcommand{\0}{\boldsymbol{0}}
\newcommand{\I}{\boldsymbol{I}}
\newcommand{\B}{\boldsymbol{B}}
\newcommand{\BMSIR}{\boldsymbol{B}_{\mathrm{MSIR}}}
\newcommand{\inv}{^{-1}}
\DeclareMathOperator{\ind}{\bot\hspace*{-6pt}\bot}
\newcommand{\V}{\boldsymbol{V}}
\renewcommand{\v}{\boldsymbol{v}}
\renewcommand{\a}{\boldsymbol{a}}

\renewcommand{\L}{\boldsymbol{L}}
\newcommand{\Gammab}{\boldsymbol{\Gamma}}
\newcommand{\Omegab}{\boldsymbol{\Omega}}
\newcommand{\Deltab}{\boldsymbol{\Delta}}

\DeclareMathOperator{\Space}{\mathcal{S}}

\DeclareMathOperator{\Real}{\mathbb{R}}
\DeclareMathOperator{\Exp}{\mathrm{E}}
\DeclareMathOperator{\Var}{\mathrm{Var}}
\DeclareMathOperator{\Cov}{\mathrm{Cov}}
\DeclareMathOperator{\diag}{\mathrm{diag}}
\DeclareMathOperator{\argmax}{\mathrm{argmax}}

\newcommand{\txt}[1]{\texttt{#1}}
\renewcommand{\hat}[1]{\widehat{#1}}
\makeatletter
\newcommand\figcaption{\def\@captype{figure}\caption}
\newcommand\tabcaption{\def\@captype{table}\caption}
\makeatother

\begin{document}

\title{Model-based SIR for dimension reduction}
\author{Luca Scrucca \\ Universit\`a degli Studi di Perugia}
\date{\today}

\maketitle

\begin{abstract}
A new dimension reduction method based on Gaussian finite mixtures is proposed as an extension to sliced inverse regression (SIR). 
The model-based SIR (MSIR) approach allows the main limitation of SIR to be overcome, i.e., failure in the presence of regression symmetric relationships, without the need to impose further assumptions. 
Extensive numerical studies are presented to compare the new method with some of most popular dimension reduction methods, such as SIR, sliced average variance estimation, principal Hessian direction, and directional regression. 
MSIR appears sufficiently flexible to accommodate various regression functions, and its performance is comparable with or better, particularly as sample size grows, than other available methods.
Lastly, MSIR is illustrated with two real data examples about ozone concentration regression, and hand-written digit classification.

\end{abstract}

\section{Introduction}

The general aim of a regression analysis is to understand how the conditional  cumulative distribution function (cdf) $F(Y|\X$) of a response variable $Y$ varies as a set of $p$ predictors  $\X=(X_1,X_2,\ldots,X_p)\T$ varies. 
Attention is often directed to  the mean function $\Exp(Y|\X)$ and to the variance function $\Var(Y|\X)$. 
Suppose that $d \le p$ linear combinations of the predictors exist such that we can write:
\begin{equation}
F(Y|\X)=F(Y|\betab\T_1\X,\betab\T_2\X,\ldots,\betab\T_d\X)=F(Y|\B\T\X),
\label{cond-cdf}
\end{equation}
where $\B = (\betab_{1},\betab_{2},\ldots,\betab_{d})$ is a $(p \times d)$ matrix of rank$(\B)=d$. If \eqref{cond-cdf} holds, then $Y$ is independent of $\X$ given $\B\T\X$, and we write $Y\ind\X|\B\T\X$.
The \textit{structural dimension of a regression} is defined as the smallest number of distinct linear combinations of the predictors required to characterize the regression of $Y$ on $\X$.
Equivalently, we can say that the subspace $\Space(\B)$ spanned by the columns of $\B$ is the \textit{dimension-reduction subspace} (DRS) for the regression of $Y$ on $\X$. It always exists, since we can trivially set $\B=\I$ but, in this case, we do not reduce the dimension, as the aim is to reduce the dimensionality of the problem as much as possible. A minimum DRS has the property of having minimum dimension among all the DRSs for the regression of $Y$ on $\X$. It can be shown that a minimum DRS may not be unique (of course, when several of such subspaces exist, they all have the same dimension). To avoid such non-uniqueness, the \textit{central dimension-reduction subspace} (CDRS) has been defined as the intersection over all DRSs. If a CDRS exists, then it is the unique minimum DRS \citep[Chap. 6]{Cook:1998}.
Every plot of $Y$ over a CDRS is called \textit{sufficient summary plot}. If we plot $Y$ over a minimum CDRS, we obtain a \textit{minimal sufficient summary plot} which will contain all the sample information available in the data about $F(Y|\X)$.

The aim of dimension reduction methods is to estimate the central subspace without estimating, or even assuming, a response model, and without strong assumptions on the form of the dependence between $Y$ and $\X$. 
Several methods have been proposed to estimate the CDRS, such as sliced inverse regression \citep[SIR;][]{Li:1991}, principal Hessian directions \citep[PHD;][]{Li:1992}, sliced average variance estimation \citep[SAVE;][]{Cook:Weisberg:1991}, parametric inverse regression \citep[PIR;][]{Bura:Cook:2001}, directional regression \citep[DR;][]{Li:Zha:Chiaromonte:2005} and inverse regression estimation \citep[IRE;][]{Cook:Ni:2005}.
They are all powerful premodeling tools for reducing high-dimensional regression problems by identifying a few linear combinations of the original predictors. 
When the structural dimension of the regression is 1, 2 or perhaps 3, as in most practical applications, the reduced dimensionality allows for effective visualization of data, and also greatly facilitates model building, particularly for non-parametric modeling.

In this paper, we propose a new dimension reduction method based on finite Gaussian mixture models (GMM). The proposal is an extension of SIR, which allows us to avoid the limitations of the basic SIR procedure without imposing further conditions. 
The next section presents the model-based SIR (MSIR) method, which is then illustrated with simulated data sets and its behavior compared with other dimension reduction methods. The consistency and sensitivity of MSIR are also discussed.
Section 3 deals with determining the dimensionality of the central subspace: two methods are discussed, a sequential test procedure and a BIC-type criterion. 
Section 4 analyses two real data examples: the first regards regression of ozone concentration levels on some primary pollutants and atmospheric conditions, and the second deals with the classification of hand-written digits.
The final section presents some concluding remarks.

\section{Model-based sliced inverse regression}

\subsection{Motivation}
\label{sec:motivation}

Sliced inverse regression (SIR) is one of the first and perhaps the most popular dimension reduction method. \cite{Li:1991} showed that, in certain conditions, an estimate of the basis of CDRS 
can be obtained by the first $d$ eigenvectors of the decomposition of $\Var(\Exp(\X|y))$ with respect to $\Var(\X)$. 

SIR requires the linearity condition and the coverage condition. The linearity condition concerns the marginal distribution of the predictors, i.e., $\Exp( \a\T\X|\B\T\X)$ must be linear in $\B\T\X$ for all $\a \in \Real^p$. \citet{Li:1991} emphasized that this condition is not a severe restriction, since most low-dimensional projections are close to being normal. With a fixed $d$, it holds approximately as $p \rightarrow \infty$ \citep{Hall:Li:1993}. 
In addition, the condition is required to hold only for the basis $\B$ of the CDRS. Since $\B$ is unknown, in practice it is required to hold for all possible $\B$, which is equivalent to the elliptical symmetry distribution (such as multivariate normal) of $\X$ \citep{Cook:Weisberg:1991}. In practice, transforming predictors so that they are approximately multivariate normal \citep{Velilla:1993} or reweighting \citep{Cook:Nachtsheim:1994} may help when gross non-linearities are present.

The coverage condition requires a method to recover all of the central subspace, not just part of it. In the context of \eqref{cond-cdf}, this condition is equivalent to requiring that $\Cov(Y,\X) \ne 0$ \citep{Yin:Cook:2005}.
It is well-known that SIR directions span at least a part of the CDRS (Cook, 1998, Prop. 10.1).
This because SIR gains information from the variation in the inverse mean function but fails when symmetric dependencies are present; this is a case of violation of the coverage condition.

\paragraph{Example} Let us consider the simple model $Y = X_1 + X_2^2$, where predictors $\X = (X_1,X_2,X_3,X_4)$ are sampled from $N_4(\0, \I_4)$ distribution; for the sake of simplicity, no error term is included. The true dimension reduction subspace is spanned by $(1,0,0,0)$ and $(0,1,0,0)$, but SIR can only find the first direction, since $\Exp(X_j|Y)=0$ for $j=2,3,4$.

\subsection{Method}

SIR estimation is based on the information provided by the inverse regression mean function $\Exp(\X|Y)$. In practice, for a continuous response variable, the range of $Y$ is sliced into $H$ non-overlapping slices $S_h$, $\Ysl = \{h: Y \in S_h \}$ for $h=1,\ldots,H$, so that the number of observations in each slice is approximately equal. Then, variation on slice means, $\mub_h = \Exp(\X|\Ysl=h)$ for $h=1,\ldots,H$, yields the SIR kernel matrix $\M = \Var(\Exp(\X|\Ysl))$, and SIR directions are obtained from the generalized eigendecomposition of $\M$ with respect to $\Var(\X)$.
The distribution of the data within any slice is summarized only by the within-slice means. The underlying assumption is that the distribution of the predictors is elliptical and compact. However, it may happen that the data follow a more complicated distribution, and important characteristics are lost if we do not take this into account.

A more flexible modeling approach may be pursued by using finite mixtures of Gaussian densities to approximate the distribution of the predictors within any slice, and then obtain the kernel matrix from the corresponding component means.
Let us assume that, for the $h$-th slice, the data can be described as follows:
\begin{equation}
f(\x|\Ysl=h) = f_h(\x) = \sum_{k=1}^{K_h} \pi_{hk} \phi(\x; \mub_{hk}, \Sigmab_{hk}),
\label{msir:mix}
\end{equation}
where $\phi(.)$ is the multivariate Gaussian density with mean $\mub_{hk}$ and covariance $\Sigmab_{hk}$, $\pi_{hk}$ are the mixing weights, so that $\pi_{hk} \ge 0$ and $\sum_k \pi_{hk}=1$, and $K_h$ is the number of components of the finite mixture.
The marginal distribution of the predictors is thus given by:
\begin{equation*}
f(\x) = \sum_{h=1}^H \tau_h f_h(\x) 
      = \sum_{h=1}^H \sum_{k=1}^{K_h} \omega_{hk} \phi(\x; \mub_{hk}, \Sigmab_{hk}),
\end{equation*}
where $\tau_h = \Pr(Y \in S_h)$ and $\omega_{hk} = \tau_h\pi_{hk}$ ($\omega_{hk} \ge 0$, $\sum_{h,k}\omega_{hk} = 1$) is the weight associated with the $k$-th component within slice $h$. 
The total number of mixture components is $K = \sum_{h=1}^H K_h$.

\begin{definition}
Consider the kernel matrix:
\begin{equation*}
\M = \sum_{h=1}^H\sum_{k=1}^{K_h} \omega_{hk} (\mub_{hk} - \mub)(\mub_{hk} - \mub)\T,
\end{equation*}
which is given by the covariance matrix of the between-component means and the marginal covariance matrix $\Sigmab = n\inv\sum_{i=1}^n (\x_i - \mub)(\x_i - \mub)\T$ with $\mub = \sum_h\sum_k \omega_{hk} \mub_{hk}$.
An estimate of the CDRS is the solution of the following constrained optimization:
\begin{equation*}
\argmax_{\B} \; \B\T \M \B, 
\;\text{ subject to } \B\T \Sigmab \B = \I,
\end{equation*}
where $\B \in \Real^{p \times d}$ is the spanning matrix and $\I$ is the $(d \times d)$ identity matrix.
This is solved through the generalized eigendecomposition:
\begin{equation}
\begin{split}
\M\v_j = \lambda_j\Sigmab\v_j & \quad \v_j\T \Sigmab \v_l = 1 \;\text{ if } j = l \text{, and } 0 \text{ otherwise}, \\
                        & \quad l_1 \ge l_2 \ge \ldots \ge l_d > 0.
\end{split}
\label{msir:decomp}
\end{equation}
The eigenvectors corresponding to the first $d$ largest eigenvalues provide a basis for the CDRS, $\BMSIR = [\v_1,\ldots,\v_d]$. 
There are at most $d=\min(p,K-1)$ directions which span this subspace, and these are the ones which show the maximal variation between component means.
When only one mixture component is used for each slice, i.e., $K_h=1$ for all slices $h=1,\ldots,H$, the kernel matrix of MSIR is equal to that provided by SIR.
\end{definition}

We call this approach MSIR (\textit{Model-based SIR}), so that the CDRS 
is spanned by directions $\BMSIR$, and the projections onto the subspace are defined as $\Z = \B\T_{\mathrm{MSIR}}\X$.

\begin{proposition}
Each eigenvalue of the eigendecomposition in \eqref{msir:decomp} is given by the variance of the between-component means along the corresponding direction of the projection subspace, i.e.
\begin{equation*}
\lambda_j = \Var(\Exp(Z_j|\Ysl^*)), \qquad\forall\ j=1,\ldots,d,
\end{equation*}
where $\Ysl^* = \{ k: Y \in S^*_k \}$, $S^*_k$ being the set made up of mixture components within each slice ($k = 1, \ldots, K)$. 
\end{proposition}

\begin{proof}
For any kernel matrix, we may rewrite the eigendecomposition in \eqref{msir:decomp} as $\M\V = \Sigmab\V\L$. Since by definition $\V\T\Sigmab\V=\I$, the diagonal matrix of eigenvalues $\L = \diag(\lambda_i)$ may be expressed as:
\begin{equation*}
\L = \V\T\M\V = \V\T\Var(\Exp(\X|\Ysl^*))\V 
   = \Var(\Exp(\Z|\Ysl^*)) = \diag(\Var(\Exp(Z_j|\Ysl^*)),
\end{equation*}
where $\Z = \B\T_{\mathrm{MSIR}}\X$, are the MSIR predictors. 
Therefore, each eigenvalue is equal to the variance of the between-component means along the associated direction.
\end{proof}

Following this result, we can interpret the contribution of each direction to the estimation of the CDRS. In addition, the directions corresponding to small eigenvalues provide little or no information about differences in means within components. Formal assessment of the number of directions required to span the CDRS is discussed in Section~\ref{sec:dim}.

\subsection{Estimation}

MSIR estimation can be pursued by applying the eigendecomposition in \eqref{msir:decomp} with suitable estimates of the unknown matrices $\M$ and $\Sigmab$. The usual sample covariance matrix $\hat{\Sigmab}$ is used for the latter. An estimate $\hat{\M}$ of the kernel matrix is computed from the estimated within-slice component means $\hat{\mub}_{hk}$ ($k=1,\ldots,K_h$; $h=1,\ldots,H$) obtained by fitting the finite mixture models in \eqref{msir:mix}.

The most popular algorithm to estimate finite mixture parameters is the Expectation-Maximization (EM) algorithm \citep{Dempster:Laird:Rubin:1977}, which converges to a maximum likelihood estimate of the mixture parameters. 
In the context of finite mixture models, 
an important point is the choice of the correct model \citep[Chapter 6]{McLachlan:Peel:2000}. In our case, this amounts to choosing both the covariance structure and the number of components. 
Parsimonious parameterization of the covariance matrices for each component within slice, $\Sigmab_{hk}$, can be achieved by imposing restrictions on such geometric feature as volume, shape and orientation of the corresponding hyperellipsoids \citep{Banfield:Raftery:1993, Celeux:Govaert:1995}.
This model selection step clearly affects the estimation of means $\mub_{hk}$ and mixture proportions $\pi_{hk}$, and thus kernel matrix $\M$.

One common approach to the problem of model selection in finite mixture modeling is based on Bayesian model selection via Bayes factors. \citet{Kass:Raftery:1995} showed than an approximation to the Bayes factor can simply be computed through the Bayesian Information Criterion (BIC). 
This proved to be efficient on practical grounds, particularly for density estimation \citep{Fraley:Raftery:1998, Fraley:Raftery:2002}.
Alternatively, \citet{Biernacki:CeleuxGovaert:2000} and \citet{Biernacki:Celeux:Govaert:Langrognet:2006} discussed the use of the Integrated Complete Likelihood (ICL) criterion. 


The algorithm for MSIR estimation may be summarized as follows:
\begin{enumerate}
\setlength{\itemindent}{1mm}
\setlength{\topsep}{0mm}
\setlength{\partopsep}{0mm}
\setlength{\parsep}{0mm}
\setlength{\parskip}{0mm}
\setlength{\itemsep}{0mm}
\item Obtain a sliced version $\Ysl$ of response variable $Y$ using $H$ non-overlapping slices (this step is not needed if $Y$ has support on a finite number of points, such as a discrete or a categorical variable).
\item Fit Gaussian finite mixture models with the EM algorithm to approximate the distribution of $\X|(\Ysl=h)$ for $h=1,\ldots,H$. The number of components $K_h$ and covariance structure $\Sigmab_{hk}$ within each slice are selected by the BIC criterion.
\item Compute kernel matrix $\hat{\M}$ from the means estimated for each mixture component within slices.
\item Perform the generalized eigendecomposition of $\hat{\M}$ with respect to the sample covariance matrix $\hat{\Sigmab}$ of the predictors.
\item The corresponding eigenvectors provide an estimate of the basis of the subspace, and are indicated as $\hat{\B}_{\text{MSIR}}=(\hat{\betab}_1,\ldots,\hat{\betab}_d)$, where $\hat{\betab}_j = \hat{\v}_j/||\hat{\v}_j||$ for $j=1,\ldots,d$, i.e., each direction is scaled to have unit norm.
\end{enumerate}


\paragraph{Example (continued)} Recalling the example discussed at the end of Section~\ref{sec:motivation}, the left-hand graphs in Figure~\ref{fig:motivating_example} show the plots of the response variable vs the first two predictors, which correspond to the basis of the subspace, for a sample of $n=400$ observations. The vertical ticks at the bottom of each graph represent the slice means for $H=5$ slices. As can be seen, the slice means for the second predictor are almost equal, which is why SIR is prevented from recovering this direction. 
Instead, the MSIR method discussed here is also able to recover the second direction. 
The right-hand graphs in Figure~\ref{fig:motivating_example} show the plots of the response variable vs the first two predictors with the estimated slice components means at the bottom. Now, means along the direction of the second predictor are spread out, which enables MSIR to recover the corresponding direction.
The estimated coefficients for the basis of the subspace are $(0.033, 0.998, -0.010, -0.045)$ and $(0.999, -0.034, -0.016, -0.020)$, with corresponding eigenvalues $0.847$ and $0.623$ (those associated with the null space are $0.107$ and $0.027$). Thus, the first direction can capture the symmetric curve, and the second direction shows the linear trend. 
 
\begin{figure}[htbp]
\centering
\includegraphics[width=0.9\linewidth]{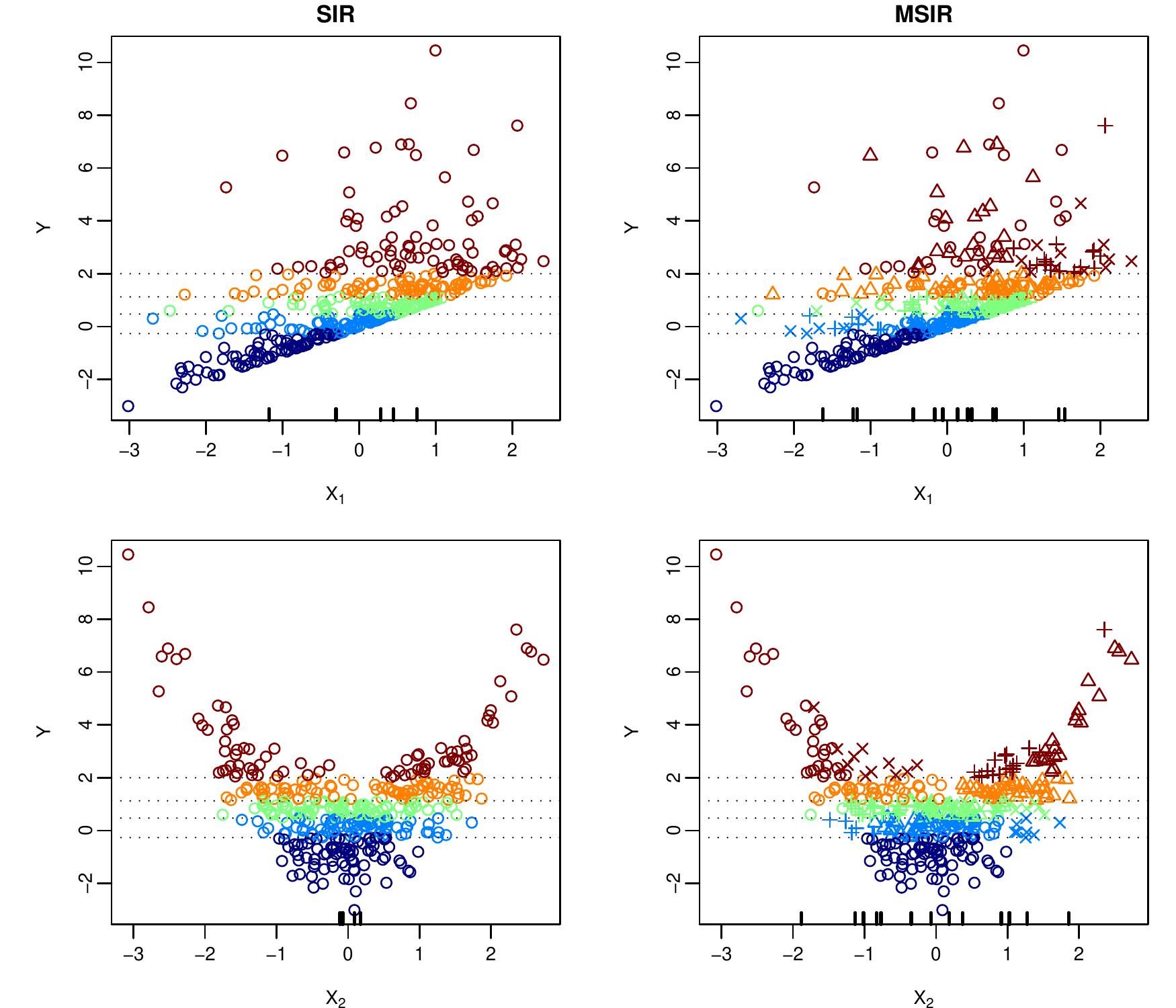}
\caption{Plots of the response variable vs the first two predictors which form the basis of the subspace for the model $Y = X_1 + X_2^2$, $\X = (X_1,X_2,X_3,X_4) \sim N_4(\0, \I_4)$. 
The horizontal dotted lines show the cutoff values used for slicing the response variable. 
Left graphs for SIR: vertical ticks at bottom of each plot represent the estimated within-slice means along corresponding direction. 
Right graphs for MSIR: points are marked by different symbols according to mixture component within-slice to which they are assigned; in this case, vertical ticks at bottom of each plot represent the estimated components within-slice means along corresponding direction.}
\label{fig:motivating_example}
\end{figure}

\subsection{Consistency of MSIR estimator}

\citet[Section~5]{Li:1991} demonstrated the $\sqrt{n}$-consistency of the SIR estimator. His arguments were based on the consistency of the individual components of the SIR algorithm.
In analogy, we argue that the MSIR estimator is $\sqrt{n}$-consistent. A full asymptotic analysis of the sample properties exceeds the scope of this paper, so this section provides a few basic ideas and results.

Let us consider the population MSIR decomposition matrix in \eqref{msir:decomp} in the equivalent form $\Sigmab^{-1/2}\M\Sigmab^{-1/2}$. $\hat{\Sigmab}$ is a $\sqrt{n}$-consistent estimator of $\Sigmab$ by the central limit theorem and, provided that $\Sigmab$ is nonsingular, $\hat{\Sigmab}^{-1/2}$ is also a $\sqrt{n}$-consistent estimator of $\Sigmab^{-1/2}$ by the continuous mapping theorem. In analogy, $\hat{\M}$ is a $\sqrt{n}$-consistent estimator of $\M$. Therefore, the eigenvectors of $\hat{\Sigmab}^{-1/2}\hat{\M}\hat{\Sigmab}^{-1/2}$ are $\sqrt{n}$-consistent estimators of the eigenvectors of the population counterpart. 

Figure~\ref{fig:msir_consistency} shows the average maximal angle between the true subspace and the subspace estimated by MSIR as a function of $1/\sqrt{n}$ for some settings of the models discussed in Section~\ref{sec:simulations1}.
If $\sqrt{n}$-consistency holds, then an approximately linear relationship should be visible in the graph, and this is the case for the examples considered. 

\begin{figure}[htb]
\centering
\def\baselinestretch{1.1}
\includegraphics[width=\linewidth]{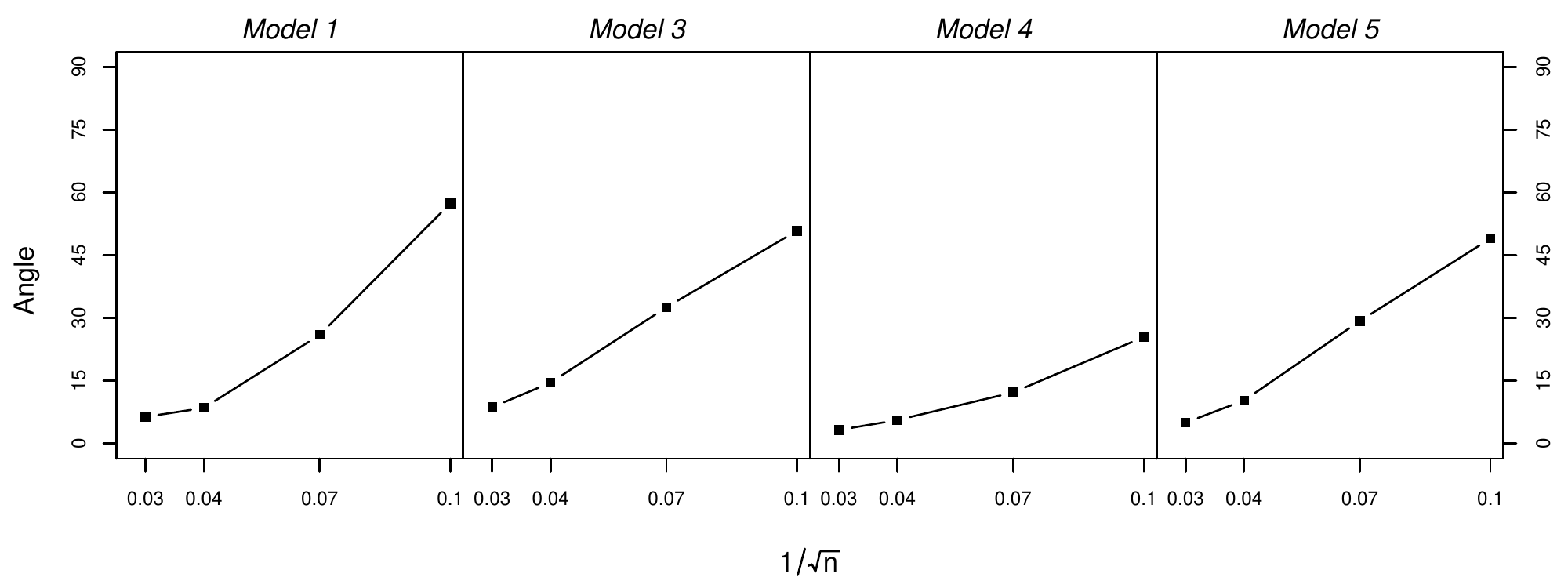}
\caption{Estimation accuracy of MSIR for checking $\sqrt{n}$-consistency. Each graph plots the average maximal angle between MSIR estimates and true subspace against $1/\sqrt{n}$. Data from simulations with $p=10$ predictors and parameters $\sigma=0.5$ for models 1 and 3, $\rho=0.5$ for model 4, $a=0.5$ for  model 5.}
\label{fig:msir_consistency}
\end{figure}

\section{Simulation studies}

\subsection{Estimation accuracy}
\label{sec:simulations1}

In this section we use simulations to examine the ability of MSIR to recover the true subspace, and compare its performance with that of other dimension reduction methods, such as SIR, SAVE, PHD and DR.
To evaluate the accuracy of a dimension reduction method to estimate the true CDRS, we made use of the following distance measure \citep[see also][]{Li:Zha:Chiaromonte:2005}. Let $\Space(\B)$ and $\Space(\hat{\B})$ be two $d$-dimensional subspaces of $\Real^p$, spanned respectively by true basis $\B$ and an arbitrary estimate $\hat{\B}$. Also let $P_{\Space(\B)}$, $P_{\Space(\hat{\B})}$ be the corresponding orthogonal projections onto $\Space(\B)$ and $\Space(\hat{\B})$.
These subspaces may be compared through the following measure:
\begin{equation}
\label{spectral_norm}
\Delta(\hat{\B}, \B) = \| P_{\Space(\hat{\B})} - P_{\Space(\B)} \| 
  = \|\hat{\B}(\hat{\B}\T\hat{\B})^{-1}\hat{\B}\T - \B(\B\T\B)^{-1}\B\T \|,
\end{equation}
where $\|.\|$ is the spectral Euclidean norm, i.e., the maximum singular value \citep{Gentle:2007}.
Equation~\eqref{spectral_norm} measures maximal angle $\alpha$ between two subspaces of $\Real^p$. It can be shown that $\Delta(\hat{\B}, \B) = \sin\alpha \in [0,1]$ \citep[p. 455]{Meyer:2000}.

In the following, we treat dimension $d$ of the CDRS as fixed. Only some results are shown here (tables and graphics of the complete simulation study appear in the Supplementary material). 

\smallskip\noindent\textbf{Model 1.}
Consider the following single-index model with a symmetric response curve:
\begin{equation*}
Y = (0.5\beta\T\X)^2 + \sigma \epsilon,
\end{equation*}
where $\beta=(1,-1,0,\ldots,0)\T$, and the predictors and the error term follow independent standard normal distributions.
It is known that one of the major limitations of SIR arises from the presence of symmetric response curves. 
The left-hand graph in Figure~\ref{fig1:sim_data3} shows a scatterplot of the response variable vs the first estimated SIR direction for a sample of $n=200$ observations on $p=5$ predictors, and $\sigma=0.1$. The curved mean function is completely absent along this projection. Conversely, the direction estimated by MSIR is shown in the right-hand graph, and the symmetric relationship with the response variable is clearly visible: note that $\Delta(\hat{\beta}_{\text{MSIR}}, \beta) = 0.086$, which corresponds to an angle of $4.9^{\circ}$, compared with an angle of $88^{\circ}$ for SIR.

\begin{figure}[htb]
\centering
\def\baselinestretch{1.1}
\includegraphics[width=0.8\linewidth]{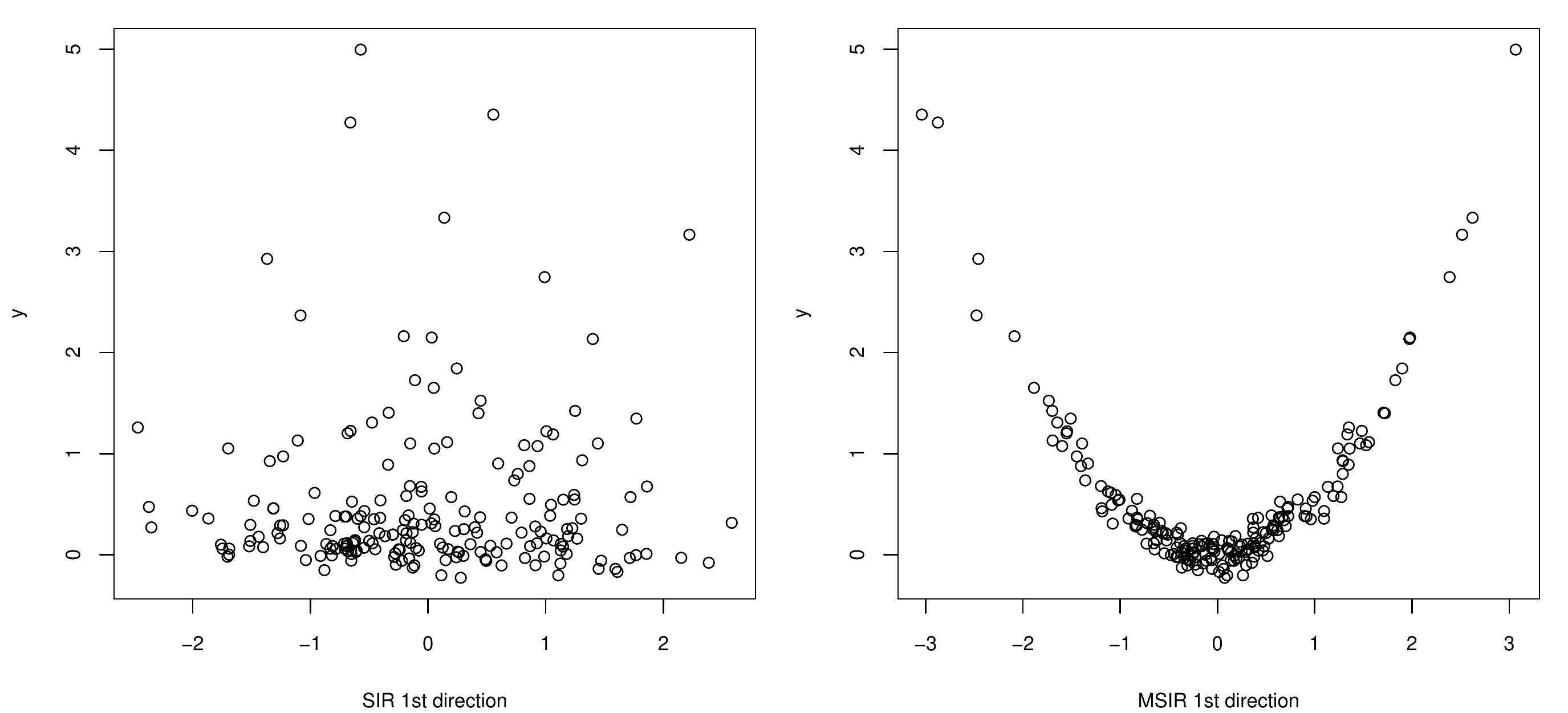}
\caption{Summary plots for single-index regression model with symmetric curve: response variable is plotted against first SIR direction (left) and first MSIR direction (right).}
\label{fig1:sim_data3}
\end{figure}

MSIR seems to be a great improvement over SIR, but it is also interesting to compare its behavior with other dimension reduction methods, particularly PHD and SAVE, which were developed to deal with such a situation. 
Figure~\ref{fig:sim_data3} shows the results of a simulation study for the above symmetric response model with number of predictors $p=10$ at various sample sizes ($n$) and error standard deviations ($\sigma$). Overall, MSIR is a great improvement over SIR. 
Compared with SAVE and PHD, which are known to work particularly well in the case of symmetric and curved relationships, the accuracy of MSIR is comparable when $\sigma$ is small and as sample size increases. When a large amount of noise is present and sample size is relatively small, MSIR tends to perform slightly less well. However, for less noisy data, the accuracy of MSIR is higher than with SAVE, PHD and DR. Note that for this model the accuracy of DR is very similar to that of SAVE.

\begin{figure}[htb]
\centering
\def\baselinestretch{1.1}
\includegraphics[width=\linewidth]{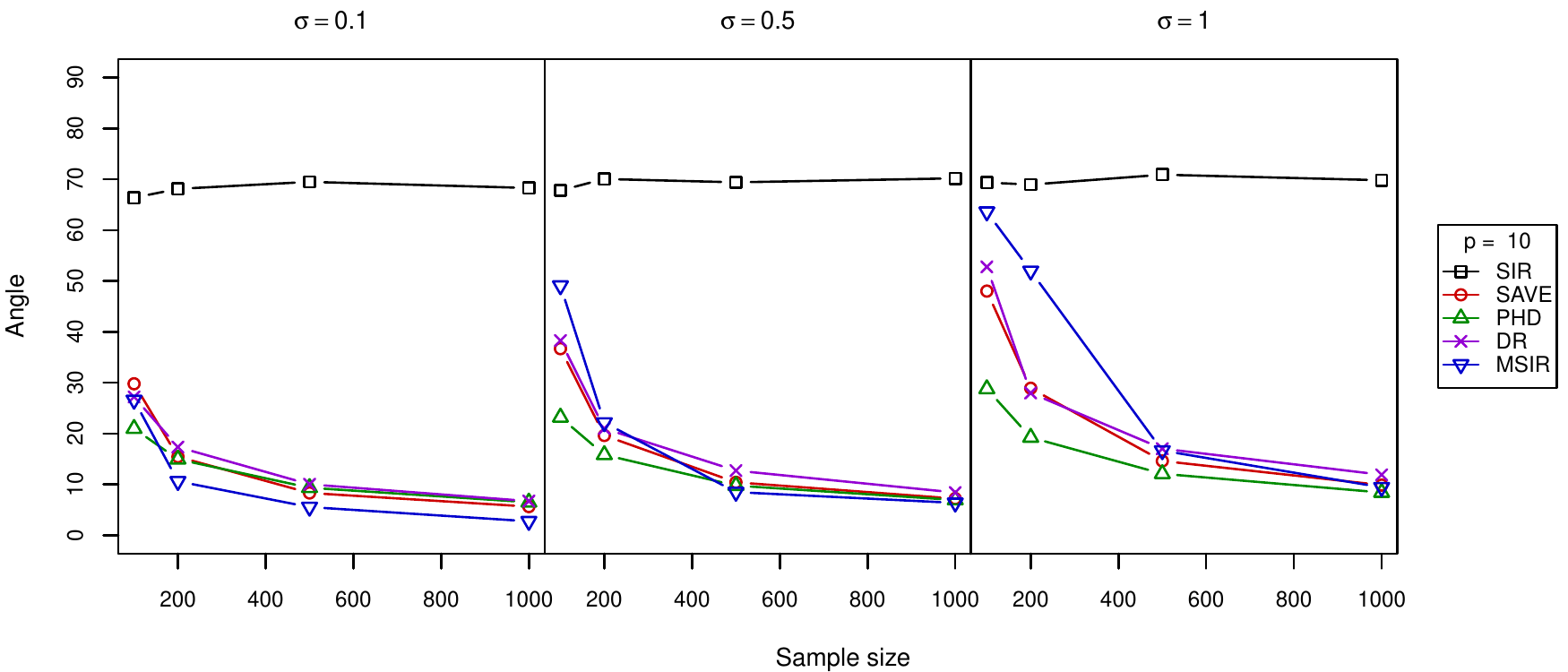}
\caption{Simulation results for \textit{model 1}: average maximum angle between true and estimated subspaces based on $500$ simulations for $p=10$ predictors at different sample sizes ($n$) and error standard deviation ($\sigma$).}
\label{fig:sim_data3}
\end{figure}

\smallskip\noindent\textbf{Model 2.}
Consider the two-dimensional regression model
 \begin{equation*}
Y = \beta\T_1\X + (\beta\T_2\X)^2 + \sigma\epsilon,
\end{equation*}
where $\beta_1 = (1,0,\ldots,0)\T$, $\beta_2 = (0,1,0,\ldots,0)\T$, and the predictors and the error term follow independent standard normal distributions. This model has both a linear trend and a symmetric quadratic curve along two different directions. We expect SIR to be able to recover the first direction but not the second, whereas the opposite is expected for PHD. SAVE and DR should be able to recover both directions, but with a different degree of efficiency. 

Figure~\ref{fig:sim_data17} shows the results from a simulation study based on 500 repetitions for each combination of sample sizes ($n$) and error standard deviation ($\sigma$), with number of predictors $p=10$. Clearly, MSIR outperforms SIR and PHD in all these settings. Its accuracy is comparable to SAVE when $p=5$ (see Supplementary material) but, as $p$ increases, MSIR is much better than SAVE. The behavior of MSIR and DR are comparable, although DR tends to provide slightly better accuracy for small sample sizes, whereas MSIR tends to achieve better accuracy as sample size grows. 

\begin{figure}[htb]
\centering
\def\baselinestretch{1.1}
\includegraphics[width=\linewidth]{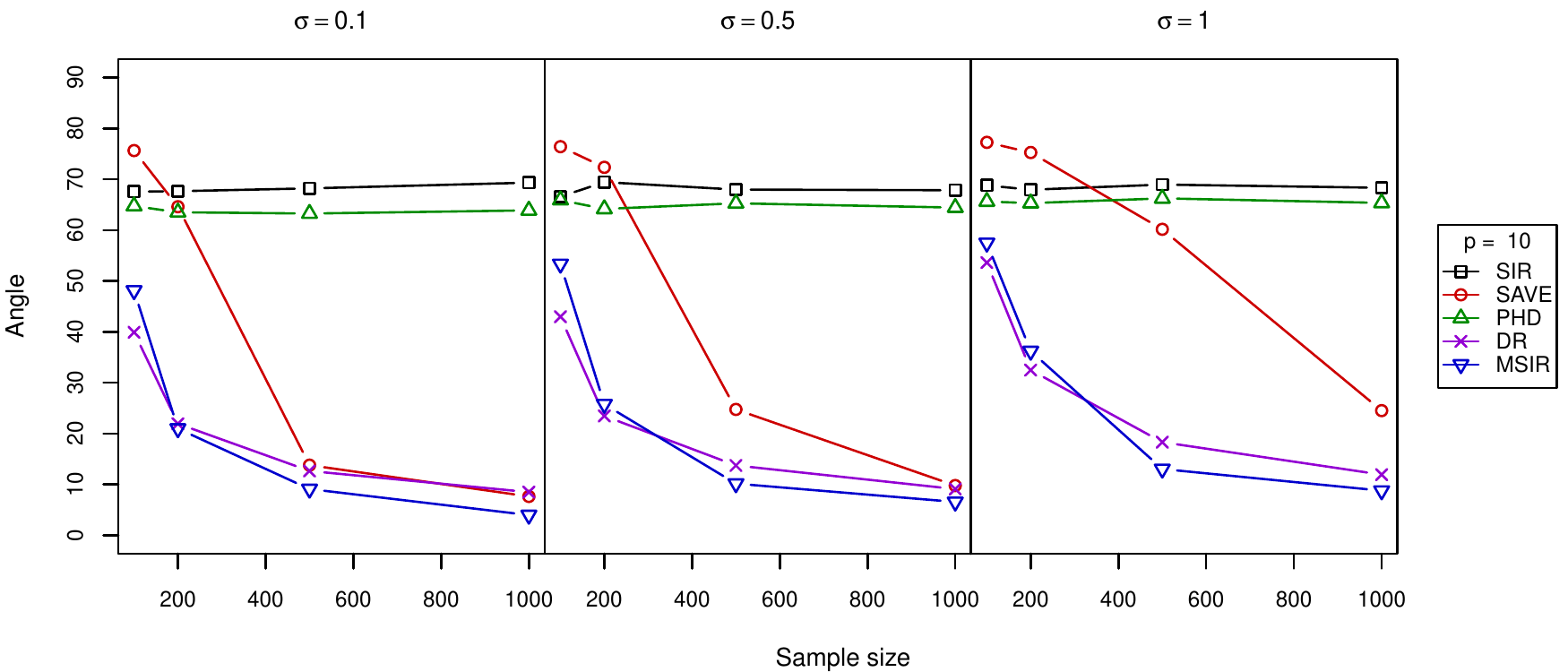}
\caption{Simulation results for \textit{model 2}: average maximum angle between true and estimated subspaces based on $500$ simulations for $p=10$ predictors at different sample sizes ($n$) and error standard deviation ($\sigma$).}
\label{fig:sim_data17}
\end{figure}

\smallskip\noindent\textbf{Model 3.}
Consider the following two-dimensional model with response rational function:
\begin{equation*}
Y = \frac{\beta\T_1\X}{0.5 + (1.5 + \beta\T_2\X)^2} + (1 +\beta\T_2\X)^2 + \sigma\epsilon,
\end{equation*}
where $\beta_1 = (1,0,\ldots,0)\T$, $\beta_2 = (0,1,\ldots,0)\T$, with the predictors and the error term which follow independent standard normal distributions. The response surface for this model shows a noisy linear trend along the first direction and a strong non-symmetric curve along the second direction. 

The simulation results of some dimension reduction methods are shown in Figure~\ref{fig:sim_data18}. Overall, MSIR is slightly, but uniformly, more accurate than SIR or DR, which behave similarly, and it is much more accurate than SAVE and PHD. MSIR, SIR and DR all improve as sample size increases, and the same happens for SAVE, except when the number of predictors is large ($p=20$, see Supplementary material). PHD performs quite badly for this data-generating model. 

\begin{figure}[htb]
\centering
\def\baselinestretch{1.1}
\includegraphics[width=\linewidth]{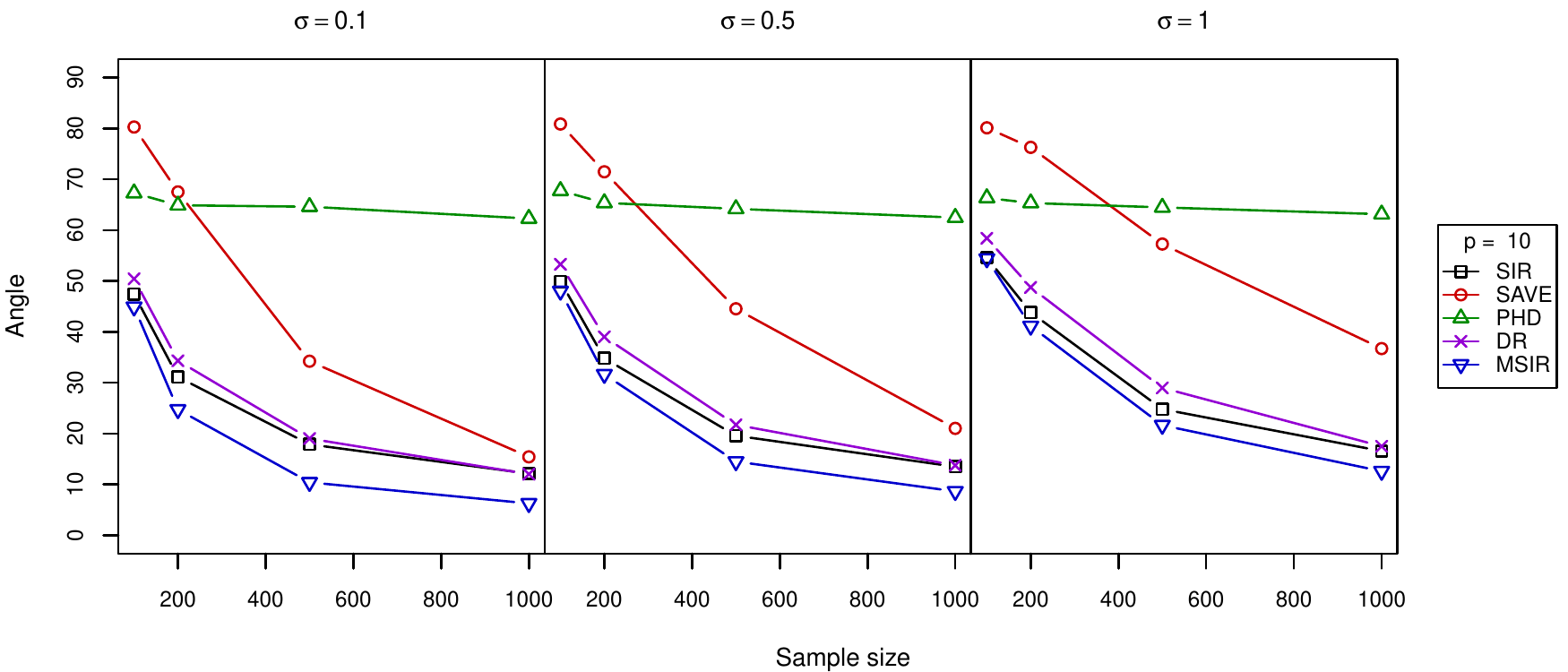}
\caption{Simulation results for \textit{model 3}: average maximum angle between true and estimated subspaces based on $500$ simulations for $p=10$ predictors at different sample sizes ($n$) and error standard deviation ($\sigma$).}
\label{fig:sim_data18}
\end{figure}

\smallskip\noindent\textbf{Model 4.}
To investigate the performance of the MSIR estimator in the case of correlated predictors we consider the following response model:
\begin{equation*}
Y = 2\beta\T\X + (\beta\T\X)^2 + \epsilon,
\end{equation*}
where $\beta = (1, 1, 1, 0, \dots, 0)\T$ and $\epsilon \sim N(0,1)$, independent of covariates. Predictors vector $\X = (X_1, \dots, X_p)$ follows a standard multivariate normal distribution with correlation between $X_i$ and $X_j$ given by $\rho^{|i-j|}$. 

Simulation results are shown in Figure~\ref{fig:sim_data22}. In general, we note that MSIR is uniformly more accurate, i.e., it always achieves a smaller angle with the true subspace than the other dimension reduction methods. When the predictors are uncorrelated ($\rho=0$), SIR, PHD and DR all provide comparable accuracy, whereas SAVE quickly deteriorates as the number of predictors increases (see Supplementary material). 
As the correlation among predictors increases, the improvement of MSIR with respect to the other methods becomes larger.
DR and PHD show similar behavior, but SIR and SAVE appear to be the least efficient methods if highly correlated predictors are present.

\begin{figure}[htb]
\centering
\def\baselinestretch{1.1}
\includegraphics[width=\linewidth]{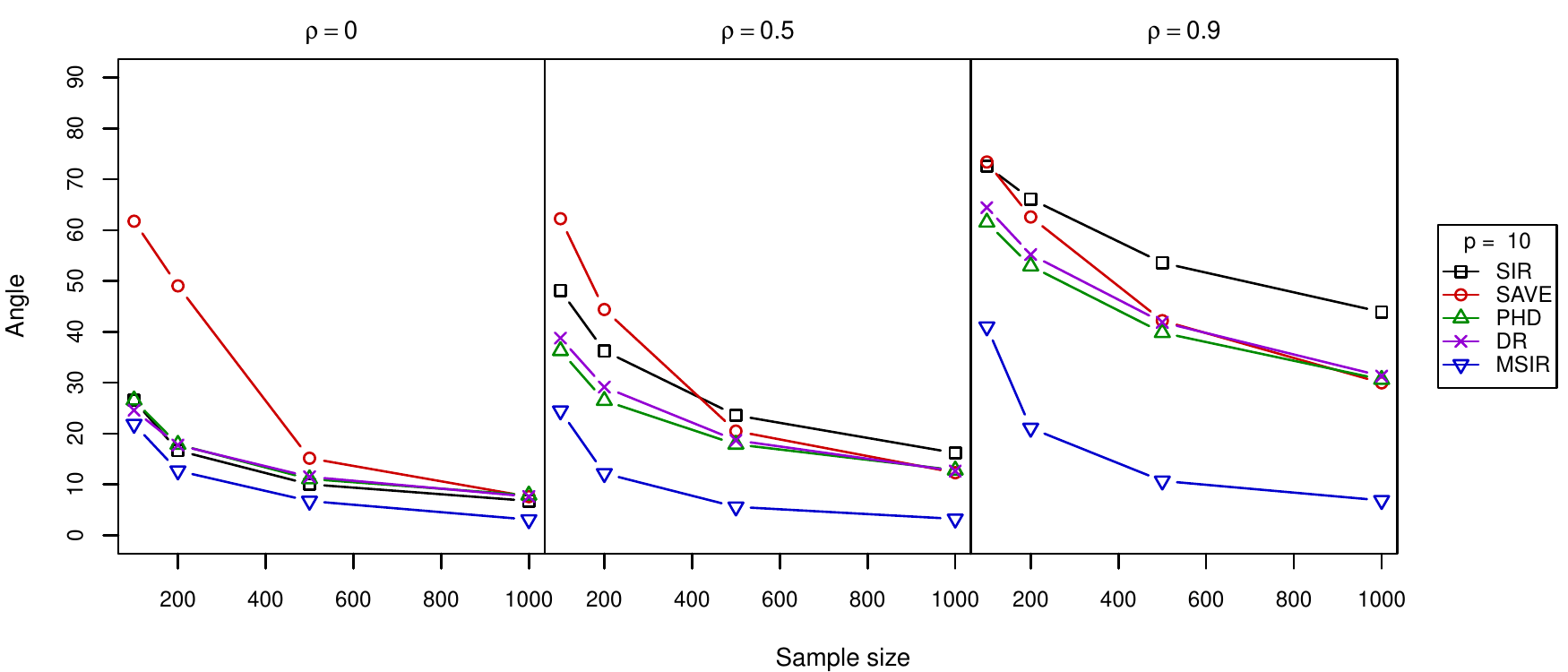}
\caption{Simulation results for \textit{model 4}: average maximum angle between true and estimated subspaces based on $500$ simulations for $p=10$ predictors at different sample sizes ($n$) and correlation coefficient ($\rho$).}
\label{fig:sim_data22}
\end{figure}

\smallskip\noindent\textbf{Model 5.}
We now consider the model discussed by \citet[Example 6.5]{Li:Zha:Chiaromonte:2005}, i.e.,
\begin{equation}
Y = \frac{1}{2}(\beta\T\X - a)^2 \epsilon,
\end{equation}
where  $\beta = (1, 0, \dots, 0)\T$, $\X \sim  N(0, I_{10})$ and $\epsilon \sim N(0,1)$, independent of predictors. Here, only the variance of $Y$ depends on the predictors and, in particular, it is a quadratic function of $X_1$ centered on values $a = \{0, 0.5, 1\}$. 
Since PHD is not capable of estimating a direction which only appears in the variance function \citep{Cook:Li:2002}, we expect PHD to perform poorly for this model. This should also happen for SIR when $a=0$, since in this case the function is symmetric around the origin. 

Figure~\ref{fig:sim_data29} shows the results of a simulation study based on 500 replications. When $a=0$, SAVE and DR perform very similarly, whereas MSIR improves as sample size increases, achieving the smallest angle when $n \ge 500$. As expected, in this case, neither SIR or PHD can estimate the true subspace. When $a$ increases to $0.5$, the performance of SAVE worsens and DR achieves the smallest angle for small sample sizes. MSIR closely follows DR and, again, it appears to be the best method for large sample sizes. In this case, SIR greatly improves with respect to the previous case, but PHD does not improve at all. When $a=1$, SIR achieves the best performance for small samples, very closely followed by MSIR and then by DR. SAVE needs large sample sizes to achieve comparable accuracy, and PHD is still the worst method.
Overall, we note that, provided that sample size is moderate to large, MSIR can provide an accurate estimate of the dimension reduction subspace in different settings when the dependence only appears in the variance function.

\begin{figure}[htb]
\centering
\def\baselinestretch{1.1}
\includegraphics[width=\linewidth]{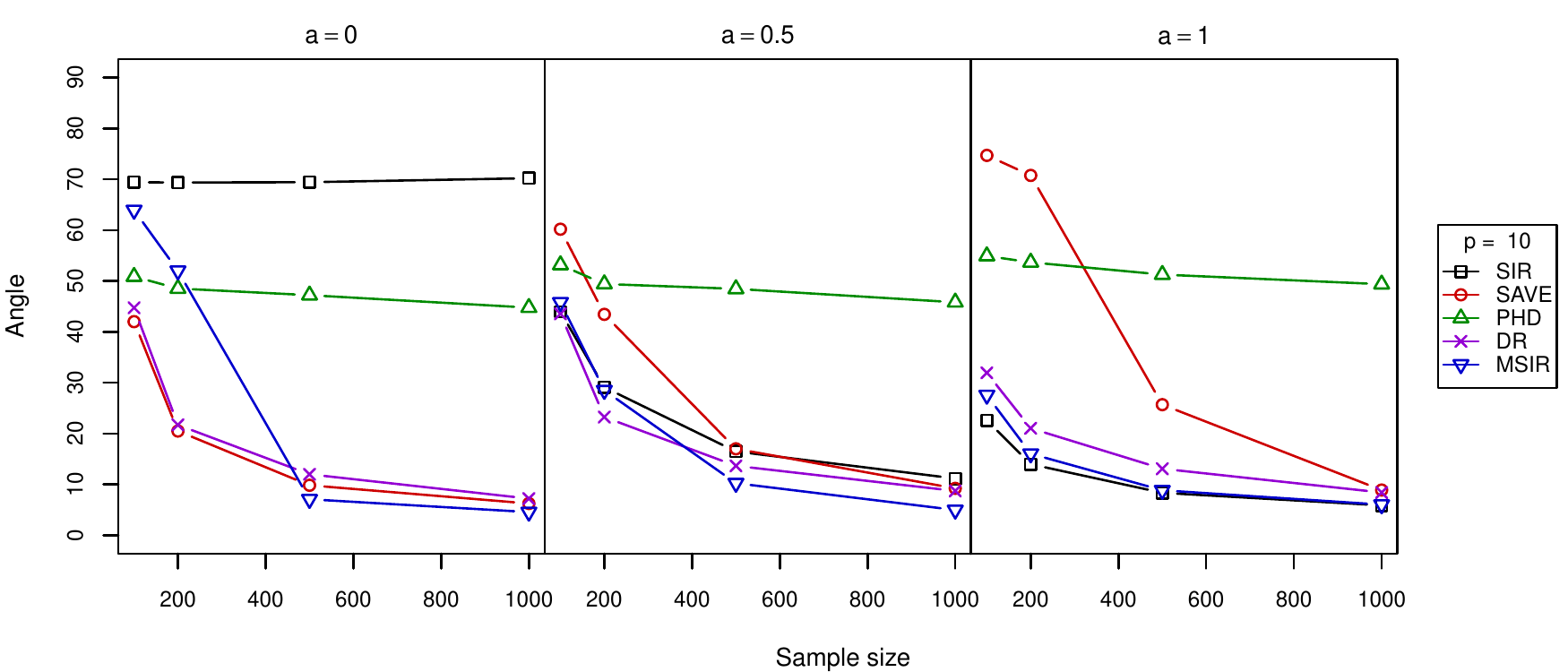}
\caption{Simulation results for \textit{model 5}: average maximum angle between true and estimated subspaces based on $500$ simulations for $p=10$ predictors at different sample sizes ($n$) and constant $a$.}
\label{fig:sim_data29}
\end{figure}

\subsection{Sensitivity of MSIR algorithm to number of slices}
\label{sec:numberofslices}

The number of slices acts as a tuning parameter, like the span width or kernel bandwidth in smoothing approaches.
Estimation of MSIR, like that of SIR, is not overly sensitive to the choice of the number of slices. However, we must ensure a sufficient number of observations within any slice to fit finite mixture models. By default, we use $H=\max(3,\lfloor\log_2(n/\sqrt{p})\rfloor)$ number of slices, where $\lfloor u \rfloor$ indicates the largest integer not greater than $u$. The resulting number of slices depends on both the amount of data available and the dimension of the predictor space (see Figure~\ref{fig:msir_nslices}). In order to have a large number of slices, we need either a large sample or a small number of predictors; for a fixed number of predictors, the number of slices increases as sample size increases. 

\begin{figure}[htb]
\centering
\def\baselinestretch{1.1}
\includegraphics[width=0.6\linewidth]{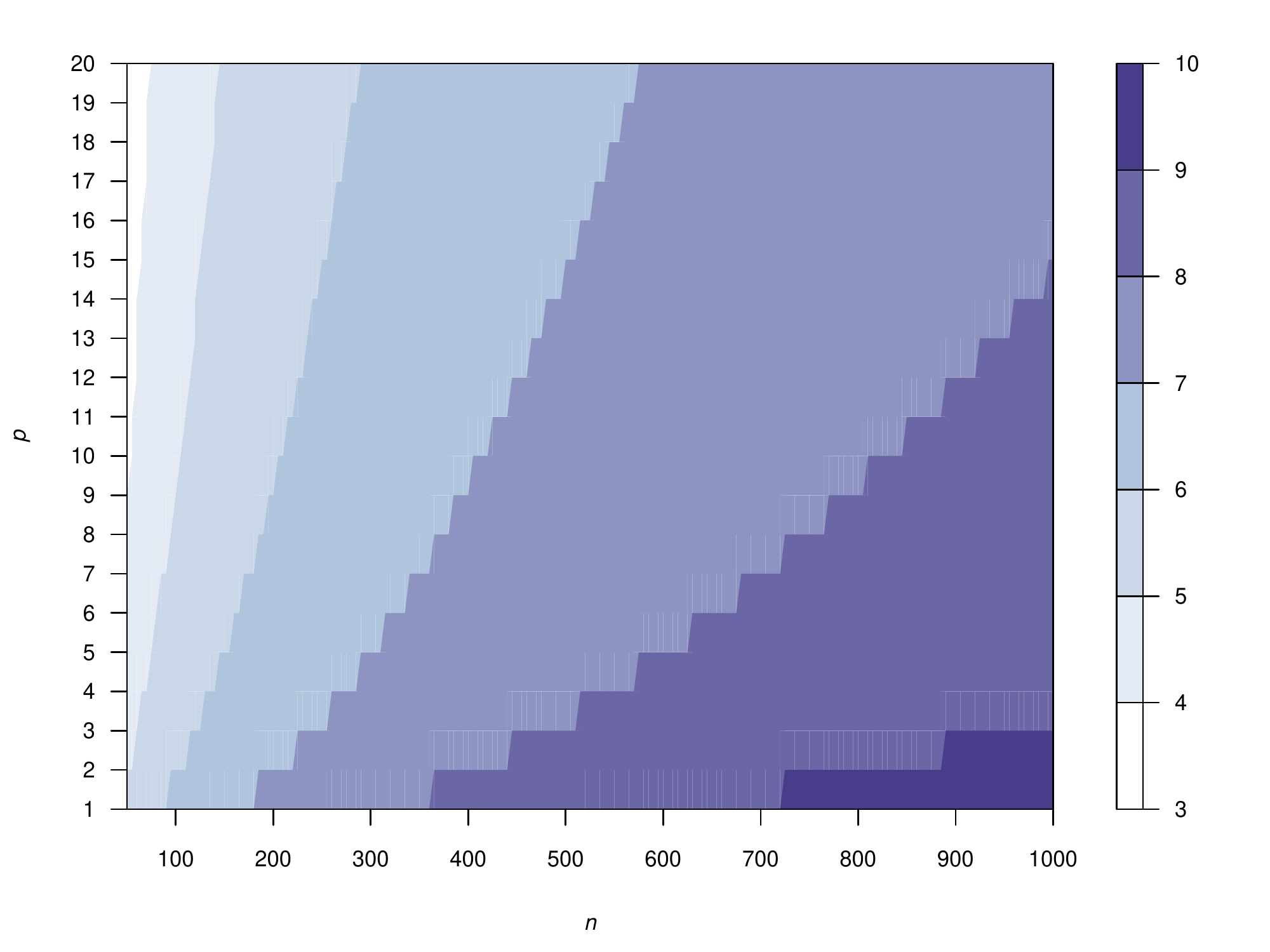}
\caption{Default number of slices used in MSIR algorithm as a function of sample size $n$ and number of predictors $p$.}
\label{fig:msir_nslices}
\end{figure}

One natural concern involves the sensitivity of the MSIR algorithm with respect to tuning parameter $H$. 
To address this issue, a simulation study was conducted in which, for models 1--4 described in Section~\ref{sec:simulations1}, we assessed the ability of MSIR to recover the true subspace when both sample size and number of slices vary. 
We set $p=5$ for the first three models with $\sigma=0.1$, and $p=10$ with $\rho=0.5$ for model 4.

Figure~\ref{fig:sim_data_h} shows the results of this simulation study. 
In general, the behavior of MSIR is quite stable, as long as we allow for enough observations within slices. 
For the first model, when $n=100$, the distributions are similar up to $H=5$, and over $H>5$ the angles become very large. When $n=200$, the break-point is at $H=10$, but is at $H=20$ when $n=500$, and at a value larger than 30 for samples of size $n=1000$. 
These characteristics are also found in the results for the second and fourth models, the third model shows a more stable distribution across values of $H$. 
Figure~\ref{fig:msir_nslices} indicates that the default number of slices is $H = (5, 6, 7, 8)$ when, respectively, $n = (100, 200, 500, 1000)$ for the first three models, and $H = (4, 5, 7, 8)$ for the last model. These values are shown as vertically shaded bars in Figure~\ref{fig:sim_data_h}, and seem to provide reasonable defaults.

\begin{figure}[htb]
\centering
\def\baselinestretch{1.1}
\includegraphics[width=0.9\linewidth]{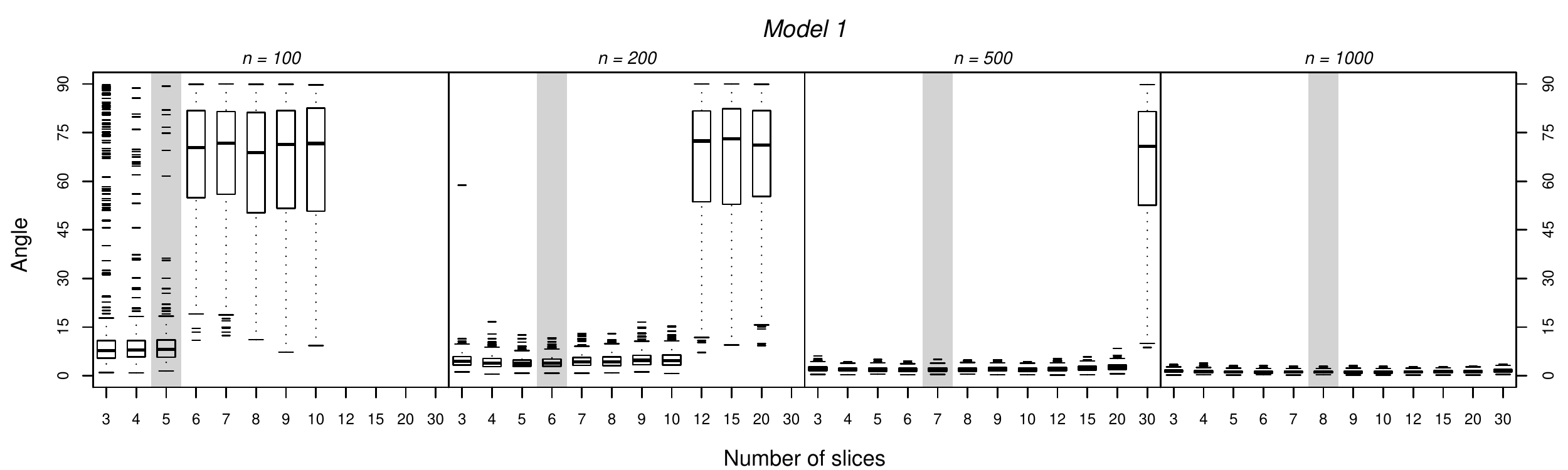}
\includegraphics[width=0.9\linewidth]{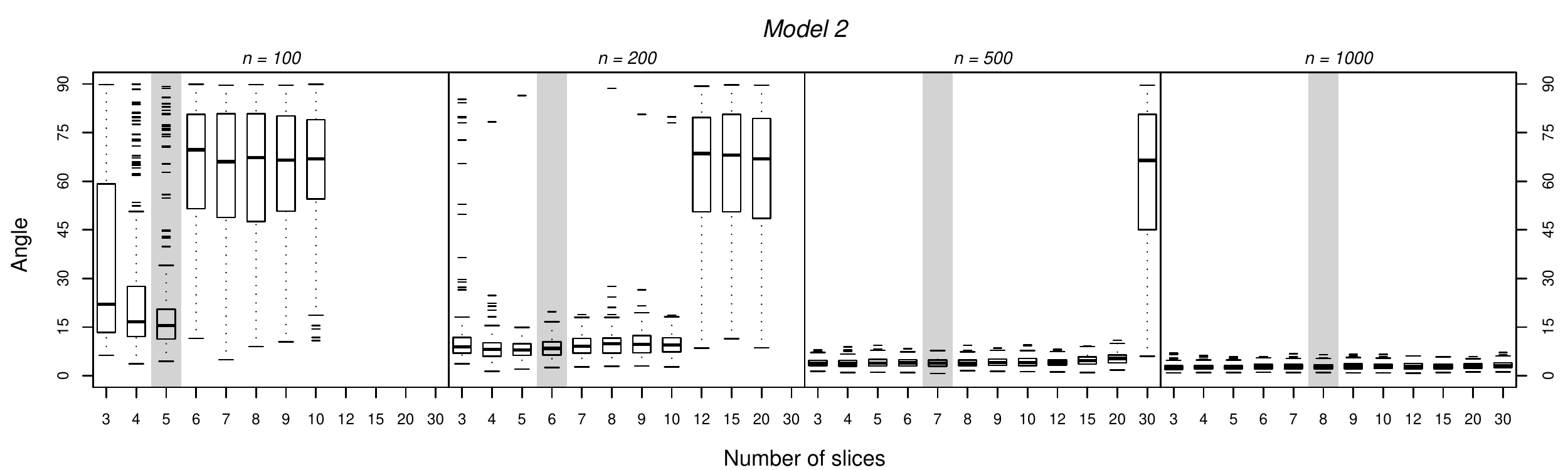}
\includegraphics[width=0.9\linewidth]{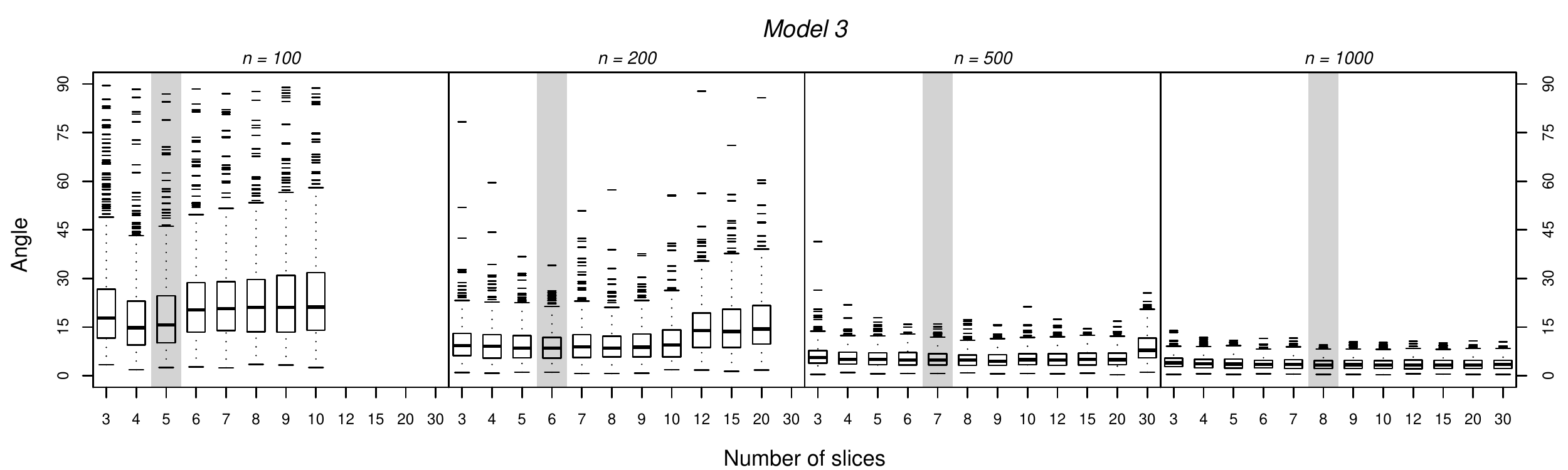}
\includegraphics[width=0.9\linewidth]{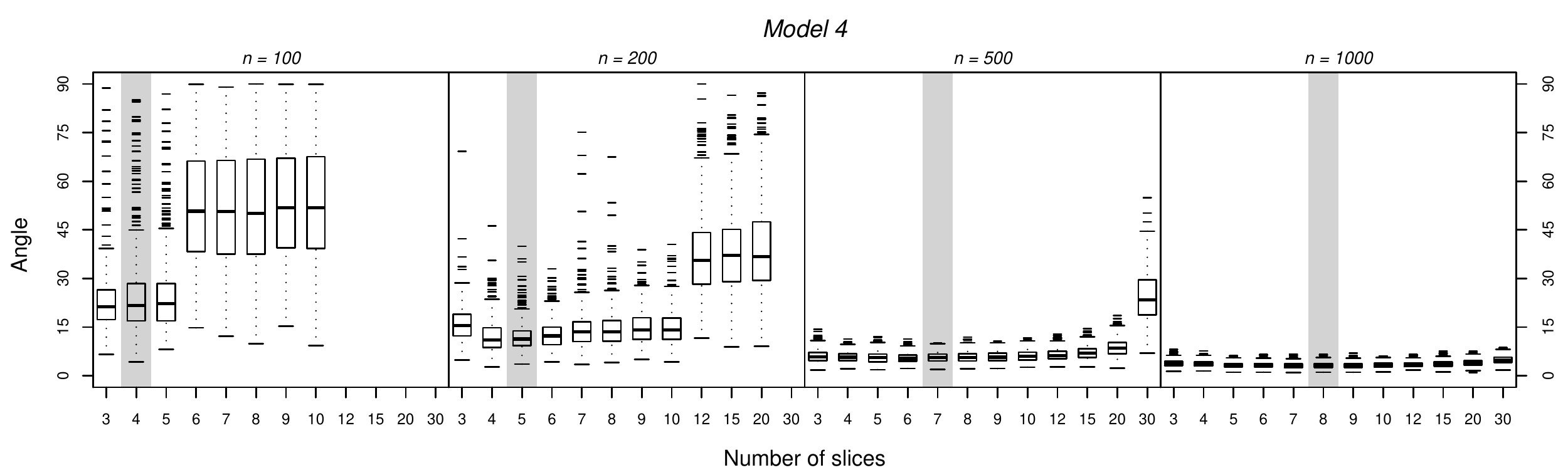}
\caption{Angles (degrees) between true subspace and MSIR estimated subspaces for model discussed in Section~\ref{sec:simulations1} as a function of number of slices ($H$) with increasing sample sizes ($n$). 
Compare with default values for $H$, in Figure~\ref{fig:msir_nslices},  represented by vertical shading bars.}
\label{fig:sim_data_h}
\end{figure}

\subsection{Computing time}

Table~\ref{tab:computing_time} gives the CPU times (in seconds) required by MSIR and other dimension reduction methods for data generated from Model 1 in Section~\ref{sec:simulations1} with different numbers of predictors ($p$) and sample sizes ($n$).
The calculations are performed in \textsf{R} \citep{RStatSoft} with a 2.2 GHz Intel Core 2 Duo Macbook Pro with 2GB RAM. Clearly, MSIR needs more computing time than the other methods, particularly as sample size increases. This is mainly because MSIR needs to estimate several mixture models via the EM algorithm and to perform model selection within each slice, in order to choose the appropriate parameterization and number of components.

\begin{table}[htb]
\centering
\caption{Comparison of computing times (in seconds)}
\label{tab:computing_time}
\begin{tabular}{lrrrrrr}
\toprule
$p$ & $n$ & SIR & SAVE & PHD & DR & MSIR \\ 
\midrule
   &  100 & 0.012 & 0.012 & 0.008 & 0.108 &  0.186 \\ 
10 &  500 & 0.020 & 0.020 & 0.013 & 0.517 &  3.535 \\ 
   & 1000 & 0.031 & 0.031 & 0.023 & 1.025 & 24.350 \\ 
\midrule
   &  100 & 0.021 & 0.021 & 0.010 & 0.265 &  0.196 \\ 
20 &  500 & 0.037 & 0.039 & 0.024 & 1.218 &  3.964 \\ 
   & 1000 & 0.058 & 0.058 & 0.042 & 2.455 & 32.307 \\ 
\bottomrule
\end{tabular}
\end{table}

\clearpage
\section{Determination of dimension of CDRS}
\label{sec:dim}

Assessing the dimension of the CDRS is an important question in any dimension reduction method.
A plot of $Y$ versus the first few MSIR predictors $\hat{Z}_j = \hat{\betab}_{j}\T\Xobs$, where $\hat{\B}_{\textrm{MSIR}} = (\hat{\betab}_1, \hat{\betab}_2, \ldots)$, is usually very informative, but inference on 
the dimension of the CDRS is still required. A popular method is based on the sequential chi-square test proposed by \cite{Li:1991}, whereas a more recent approach is based on a BIC-type criterion. 
In this section, we discuss how to apply these two methods in the MSIR case.

\subsection{Permutation test}
\label{sec:dimpermtest}

\cite{Li:1991} proposed a sequential test procedure for SIR based on the statistic
\begin{equation}
\hat{\Lambda}_{d} = n \sum_{j=d+1}^{p} \hat{\lambda}_{j},
\label{test-stat}
\end{equation}
which, under the assumption that the predictors are normally distributed, has an asymptotic chi-square distribution with $(p-d)(H-d-1)$ degrees of freedom. 
In general, chi-square asymptotic distribution holds for any distribution of the predictors under the linearity and constant covariance conditions \citep{Bura:Cook:2001}.
For other dimension reduction methods, for instance SAVE, the null distribution of statistic \eqref{test-stat} is unknown, even asymptotically.
In these cases, and for SIR when the linearity and constant covariance conditions are not satisfied, \cite{Cook:Weisberg:1991} and \cite{Cook:Yin:2001} proposed a general permutation test which can be easily adapted to our case. 

Consider partition $\B = (\B_{1},\B_{2})$ of the $(p \times p)$ matrix of eigenvectors of population kernel matrix $\M$, where $\B_{1}=(\betab_{1},\ldots,\betab_{d})$ and $\B_{2}=(\betab_{d+1},\ldots,\betab_{p})$. Assume that the independence condition between $(Y, \B\T_{1}\Xobs)$ and $\B\T_{2}\Xobs$ holds for testing hypothesis $H_0: \rank(\M) \le d$ versus $H_1: \rank(\M) > d$.
The observed test statistic $\hat{\Lambda}_{d}$ (for $d=0,1,\ldots,p-1$) can be compared to its permutation distribution under the null hypothesis. Starting with $d=0$, the test procedure is performed sequentially. If the null hypothesis is not rejected for a given value of $d$, then the last $(p-d)$ MSIR predictors $\B\T_{2}\Xobs$ can be discarded without loss of information on the regression of $Y$ on $\Xobs$. Thus, the testing procedure involves the following steps:
\begin{enumerate}
\setlength{\itemindent}{1mm}
\setlength{\topsep}{0mm}
\setlength{\partopsep}{0mm}
\setlength{\parsep}{0mm}
\setlength{\parskip}{0mm}
\setlength{\itemsep}{0mm}
\item for a given sample kernel matrix $\hat{\M}$, compute the eigendecomposition in \eqref{msir:decomp} to obtain eigenvectors $\hat{\B}_{1}=(\hat{\betab}_{1},\ldots,\hat{\betab}_{d})$ and $\hat{\B}_{2}=(\hat{\betab}_{d+1},\ldots,\hat{\betab}_{p})$, with associated eigenvalues $\hat{\lambda}_1,\ldots,\hat{\lambda}_d$ and $\hat{\lambda}_{d+1},\ldots,\hat{\lambda}_p$;

\item compute the observed value of test statistic $\hat{\Lambda}_{d}$;

\item obtain the vectors of sample MSIR predictors $\Zest_{i1} = \hat{\B}\T_1\Xobs_i$ and $\Zest_{i2} = \hat{\B}\T_2\Xobs_i$, for $i=1,\ldots,n$;

\item randomly permute indices $i$ of $\Zest_{i2}$ to obtain permuted data $\Zest_{i^*2}$;

\item apply the MSIR procedure to original data $Y_i$, $\Zest_{i1}$ and permuted data $\Zest_{i^*2}$, to obtain the value of permuted test statistic $\hat{\Lambda}^*_{d}$;

\item repeat steps 4 and 5 a number of times. The p-value for testing the null hypothesis is estimated as the fraction of $\hat{\Lambda}^*_{d}$ exceeding $\hat{\Lambda}_{d}$.

\end{enumerate}
For $d=0,1,\ldots,p-1$, we test $\rank(\M)$ sequentially, and estimate $\hat{d} = d_0$ if $d_0$ is such that the corresponding $p$-value is the first one greater than a fixed significance level, say $\alpha=0.05$, in the series. If we reject all the hypotheses, we conclude that $\rank(\M) = p$.

\subsection{BIC-type criterion}
\label{sec:dimbic}

\cite{Zhu:Miao:Peng:2006} and \cite{Zhu:Zhu:2007} proposed a consistent BIC-type procedure to determine the dimension of the CDRS. 
Let $\Omegab = \Gammab + \I_p$ and $\hat{\Omegab} = \hat{\Gammab} + \I_p$, where $\Gammab$ is the kernel matrix for standardized predictors and $\I_p$ is the $(p \times p)$ identity matrix. 
Let $\theta_1 \ge \theta_2 \ge \ldots \ge \theta_p$ be the eigenvalues of $\Omegab$ and $\hat{\theta}_1 \ge \hat{\theta}_2 \ge \ldots \ge \hat{\theta}_p$ those of $\hat{\Omegab}$. 
Clearly, $\theta_i = \lambda_i + 1$, where $\lambda_i$ are the eigenvalues of $\Gammab$, and the dimension of the CDRS is given by the number of eigenvalues of $\Omegab$ greater than 1. 
\cite{Zhu:Miao:Peng:2006} showed that a BIC-type criterion can be defined as follows
\begin{equation*}
G(d) = \log L_d - C(n,p,d),
\end{equation*}
where $\log L_d = \frac{n}{2} \sum_{i=1+\min(\tau,d)}^{p} (\log(\hat{\theta_i}) + 1 - \hat{\theta_i})$, with $\tau$ denoting the number of $\hat{\theta}_i>1$, and $C(n,p,d)$ is a penalty term which depends on the number of free parameters to be estimated. 
In the original proposal the penalty term was defined as
$C(n,p,d) = C_n d(2p-d+1)/2$, with $C_n = (0.5\log(n) + 0.1 n^{1/3})/(2(n/H))$, where $n/H$ is the average number of data points within each slice. 
However, this definition of the penalty term was based on favorable empirical evidence among a candidate set of penalty terms. 
Later, \cite{Zhu:Zhu:2007} noted that the number of $\theta_i$ to be estimated are $(p-d)$, and suggested the use of the penalty
$C(n,p,d) = -(p-d)\log(n)$.
The dimension of the CDRS is then estimated as the maximizer of $G(d)$, i.e.
$\hat{d} = \argmax_{0 \le d \le p-1} \; G(d)$.
This BIC-type procedure for selecting the dimension of the CDRS is easily applied to the MSIR approach by setting $\hat{\Gammab} = \hat{\Sigmab}^{1/2} \hat{\M} \hat{\Sigmab}^{1/2}$.

\subsection{Simulation study}
\label{sec:simulations2}

We conducted a simulation study using the first four models described in Section~\ref{sec:simulations1} to investigate the accuracy of the permutation test (PT) procedure and the BIC-type criterion in choosing the correct dimension of the CDRS. 
Figure~\ref{fig:simdim} shows the results of these simulations, plotting fractions $F(i)$ and $F(i,j)$ based on 500 replications, in which a procedure (PT or BIC) selected $d=i$ and $d=i$ or $d=j$ versus sample size. To simplify the discussion, only the results for case $p=10$ are reported.

For the first model, which has $d=1$, the PT procedure tends to select the correct value as sample size increases. When sample size is small and there is a large amount of noise, the procedure sometimes underestimates the true dimension. The behavior of BIC is similar to that of PT, except for when $n=1000$ and $\sigma=0.1$, in which case it overestimates the dimensionality.
For the second model, which has $d=2$, the BIC-type criterion greatly improves as sample size increases, whereas the PT procedure is more accurate for small sample sizes. The noise component does not seem to affect the accuracy of either procedures. On the contrary, it has a large effect for model 3, which also has $d=2$. In this case, both PT and BIC worsen as $\sigma$ increases: in particular, they tend to select only one direction as relevant.
For the last model, where $d=1$, the PT procedure performs well, and the BIC-type criterion is comparable when the predictors are uncorrelated or very strongly correlated but, if $\rho=0.5$, it tends to overestimate the true dimensionality. 

Overall, both procedures provide reliable estimates of the dimension of the CDRS. The permutation test procedure is more accurate when sample size is not large, whereas the BIC-type criterion is more efficient as sample size increases.

\begin{figure}[htb]
\centering
\linespread{1.1}
\includegraphics[width=0.7\linewidth]{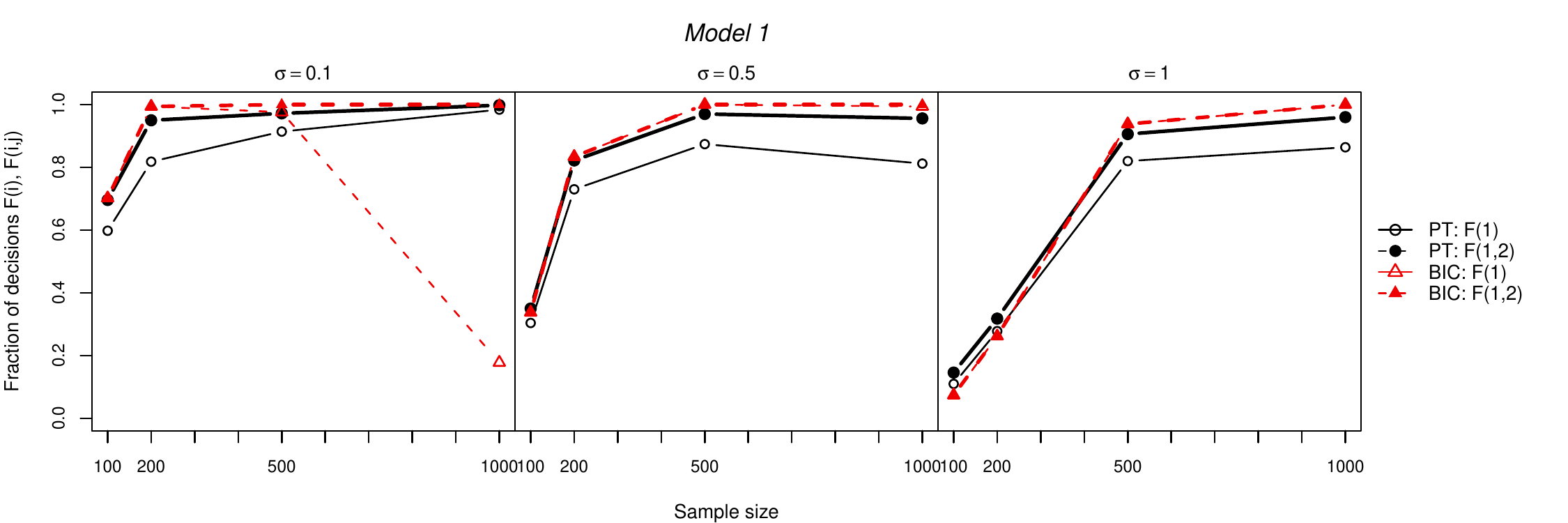}
\includegraphics[width=0.7\linewidth]{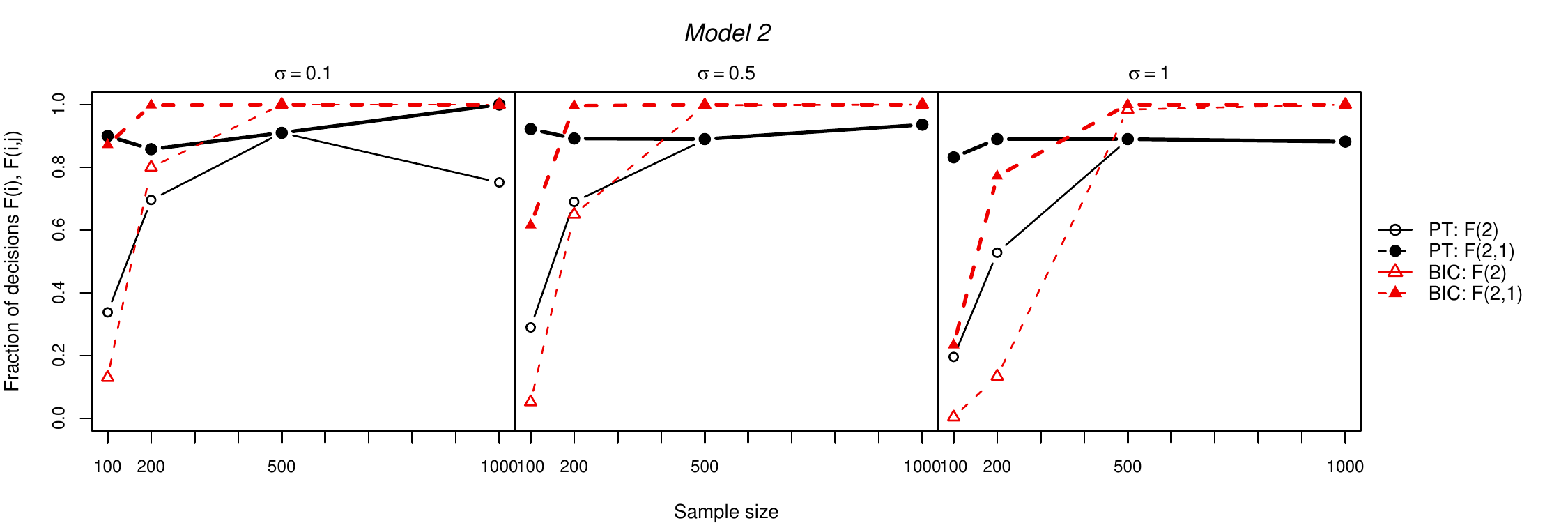}
\includegraphics[width=0.7\linewidth]{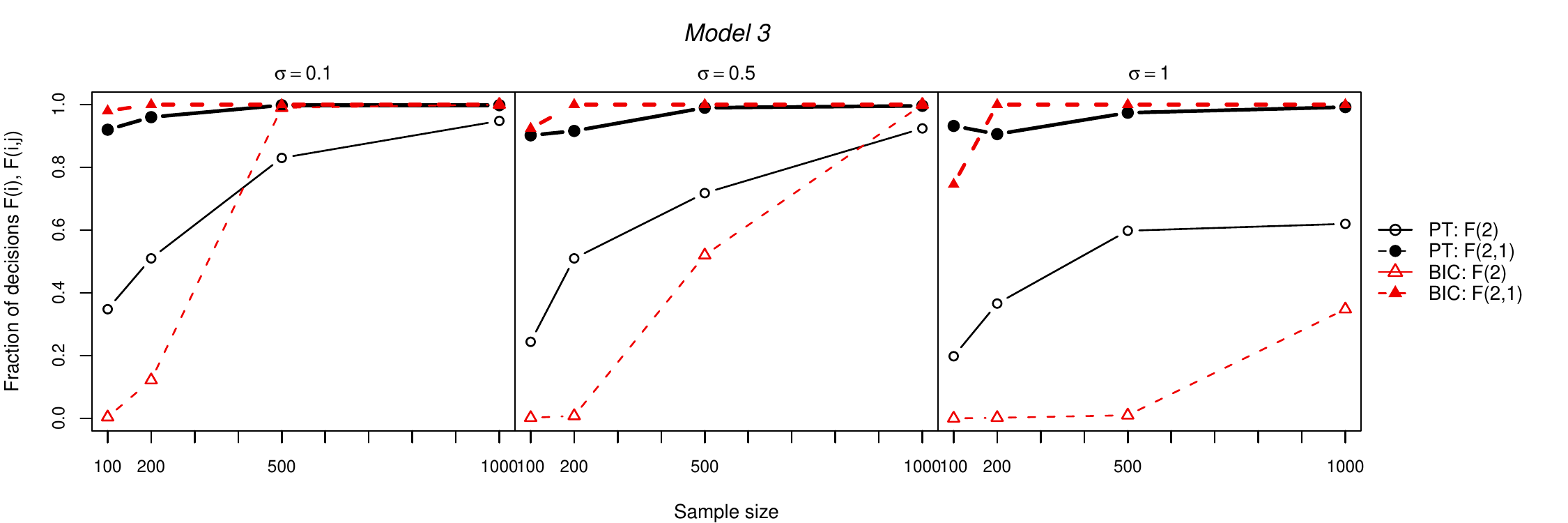}
\includegraphics[width=0.7\linewidth]{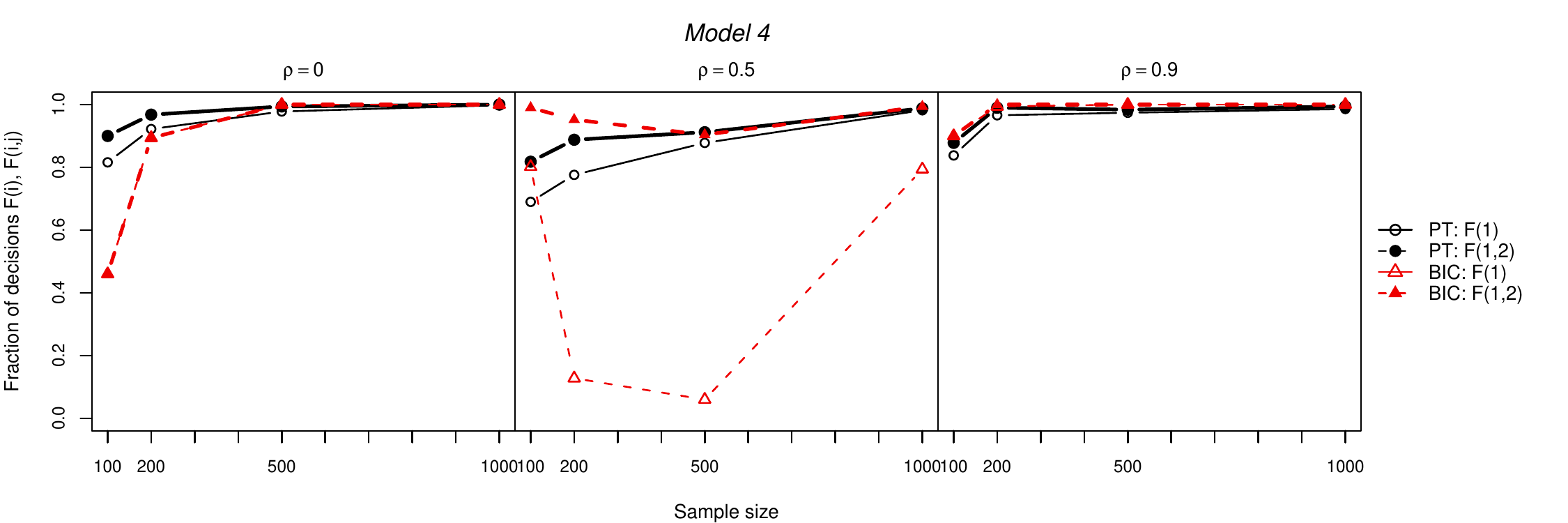}
\caption{Inference about $d$ from simulations for four models described in Section~\ref{sec:simulations1}. $F(i)$, $F(i,j)$ are fractions of runs in which estimated $d$ was one of the arguments.}
\label{fig:simdim}
\end{figure}

\clearpage

\section{Data analysis}

\subsection{Chicago air pollution data}

Atmospheric pollutants are responsible for serious environmental pollution and may have dangerous effects on public health.
Pollutants are often classified as either primary or secondary. Primary pollutants are released into the atmosphere during combustion processes of any kind (volcanic eruptions, motor vehicle exhausts, etc.), and include carbon monoxide \txt{CO}, nitrogen dioxide \txt{NO2}, sulfur dioxide \txt{SO2}, and particulate matter with diameter smaller than 10 microns \txt{PM10}.
After their release into the atmosphere, primary pollutants are subject to processes of diffusion, transport and deposition. They also undergo processes of chemical and physical transformation, which may lead to the formation of secondary pollutants. These are formed from primary pollutants as a result of changes of various kinds caused by reactions which often involve atmospheric oxygen and weather conditions. Of main interest is the ground level of ozone (O3) which, at abnormally high concentrations, caused by human activities (mainly the combustion of fossil fuel) is a dangerous pollutant.

\begin{table}[htb]
\centering
\caption{Model-based SIR results for air pollution data.}
\label{tab1:chic97airpollution}
\begin{tabular*}{0.87\textwidth}{lccccccc}
\toprule
Slices     & 1   & 2         & 3   &  4       & 5       & 6   & 7 \\ 
GMM        & XXX & EEI       & VVV &  VEI     & VEI     & XXX & XXX \\
Num. comp. & 1   & 4         & 2   &  3       & 3       & 1   & 1 \\ 
Num. obs.  & 52  & 5$|$23$|$7$|$17 & 45$|$7&  13$|$31$|$8 & 26$|$3$|$23 & 52  & 51 \\
\end{tabular*}
\begin{tabular*}{0.87\textwidth}{lrrrrrr}
\midrule
Predictors & \multicolumn{6}{c}{Standardized basis} \\
\cline{2-7}
 & Dir$_1$ & Dir$_2$ & Dir$_3$ & Dir$_4$ & Dir$_5$ & Dir$_6$ \\ 
\cline{2-7}
\txt{T}    &  0.6824 &  0.15446 &  0.00996 & -0.15674 &  0.6137 & -0.1132\\
\txt{H}    & -0.1307 & -0.07566 & -0.40980 &  0.48761 &  0.4306 &  0.2151\\
\txt{PM10} &  0.1189 & -0.48158 & -0.38248 &  0.38666 & -0.4828 & -0.5056\\
\txt{SO2}  & -0.1406 & -0.43371 & -0.44994 & -0.59750 &  0.2859 &  0.3660\\
\txt{NO2}  &  0.6204 &  0.37996 &  0.13235 &  0.47879 & -0.2149 &  0.6795\\
\txt{CO}   & -0.3136 & -0.63719 &  0.68243 & -0.04374 &  0.2775 & -0.2994\\
\midrule
\end{tabular*}
\begin{tabular*}{0.87\textwidth}{lrrrrrrr}
Eigenvalues  & 0.7381 &  0.4514 & 0.1828 & 0.1371 & 0.09066 & 0.04821 \\
Structural dimension &      0 &     1 &     2 &     3 &     4 &    5 \\  
BIC-type criterion   & -17.77 & 9.974 &  18.4 & 15.21 & 10.88 & 5.69 \\
Test statistic       & 598.4  & 330.4 & 166.5 & 100.2 & 50.41 & 17.5 \\
Permutation \textit{p-value} & 0      & 0.01  & 0.25  & 0.29  & 0.36  & 0.33 \\
\bottomrule
\end{tabular*}
\end{table}

We considered daily data collected in Chicago in 1997 and available at \url{http://www.ihapss.jhsph.edu/data/data.htm}. We aimed at modeling ozone concentration \txt{Y} on some primary pollutants and weather conditions (temperature \txt{T} and humidity \txt{H}). The results from MSIR estimation are shown in Table~\ref{tab1:chic97airpollution}: the first part of the table lists the type of GMM fitted for each slice \citep[for the meaning of symbols, see][]{Fraley:Raftery:2006Mclust}, the number of mixture components, and the number of observations for each within-slice component.
The second part of the table shows the predictor coefficients, scaled to have standard deviation equal to one, associated with the estimated directions.  
The eigenvalues of the MSIR kernel matrix are also shown, together with the BIC-type criterion and the permutation test described in Section~\ref{sec:dim}. Both methods indicate a two-dimensional structure. 

The plot of the response variable versus the first two MSIR variates are shown in Figure~\ref{fig1:chic97airpollution}, where smooth functions for mean and variance have been added as described in \citet[pp. 275--278]{weisberg:2005}.
A rotating 3D plot is also available in the Supplementary material.
An increasing trend with constant variance is associated with the first MSIR direction, which is mainly determined by predictors \txt{T} and \txt{NO2}. Thus, an increase in ozone level is associated with increasing values of temperature and nitrogen dioxide. 
The second direction shows a curved relationship, with non-constant variance. However, its interpretation is less straightforward: there is a positive relationship with \txt{T} and \txt{NO2}, as in the first direction, but an inverse relationship with the other predictors, especially \txt{PM10}, \txt{SO2} and \txt{CO}. 

When we compare the estimated MSIR directions with those obtained by other dimension reduction methods, we can see that the first MSIR variate has $R^2 \approx 0.98$ with the first SIR variate, and $0.92$ with the first DR variate. Therefore, the three methods essentially identify the same direction. 
In contrast, the second MSIR variate has an $R^2$ of about $0.5$ with the second variate estimated by both SIR and DR. Therefore, although these directions are different, they all show a heteroskedastic shape.

\begin{figure}[htb]
\centering
\def\baselinestretch{1.1}
\includegraphics[width=0.9\linewidth]{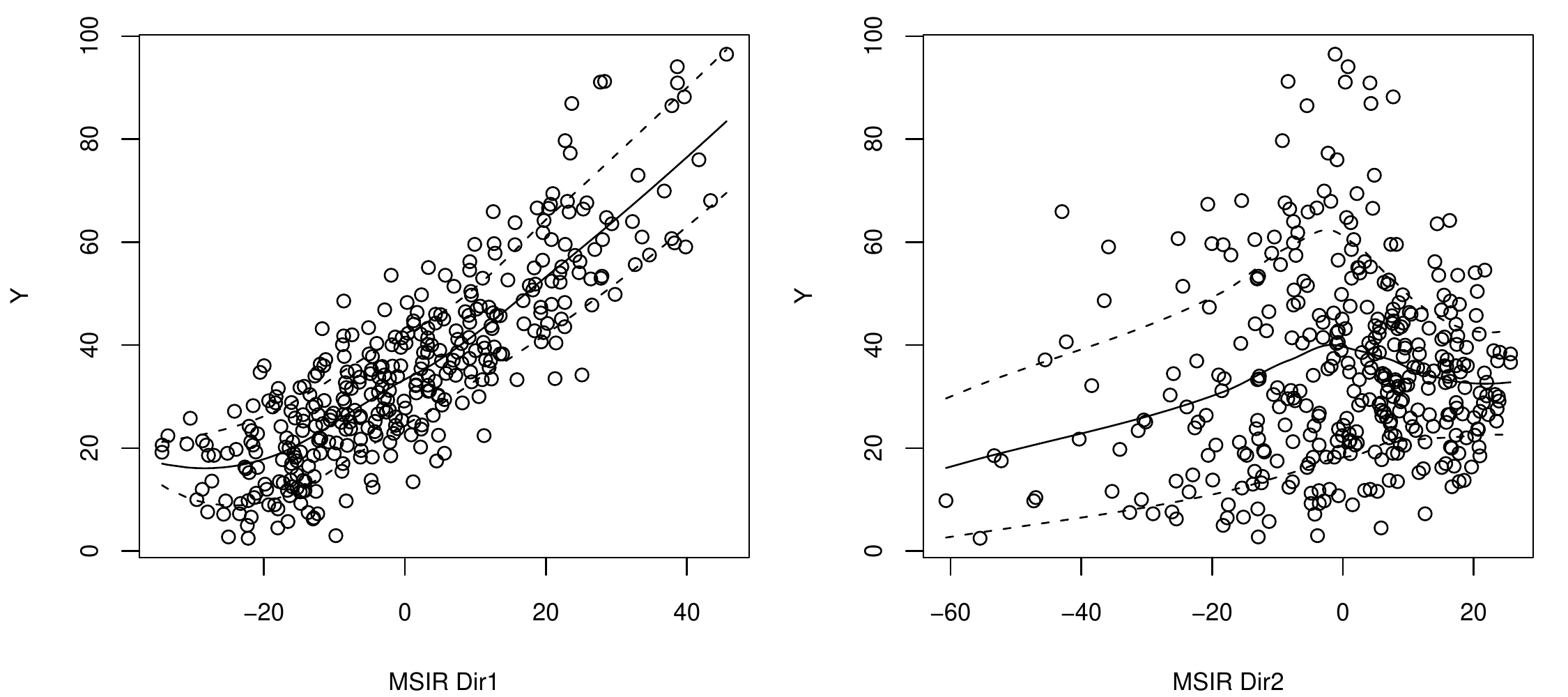}
\caption{Summary plots of ground level ozone concentration ($Y$) versus first two estimated MSIR directions, with smooth functions for mean and variance.}
\label{fig1:chic97airpollution}
\end{figure}

\subsection{Pen digit data}

The data for this pattern recognition problem on handwritten digits come from the UCI machine-learning repository and contain samples of handwritten digits $\{0,1,\ldots,9\}$ collected from 44 different writers. Each digit is stored as a 16-dimensional vector. 
The data set is divided into a training set and a learning set. We focus on the data involving three digits, \{0, 6, 9\}. Because of their similar shape, they are among the most difficult to identify. 
These data were analysed by \citet{Zhu:Hastie:2003} by means of several procedures including SIR and SAVE, and by \citet{Li:Wang:2007} with DR. 
The latter authors noted that SIR provides only locational separation of the three types of digits, whereas their DR method also provides a distinction in variation \cite[see Figure~3 of][]{Li:Wang:2007}.

For this classification problem, the response variable is the class label of each digit. We applied the proposed MSIR method to the training set made up of 2219 digits. 
For the group of 0 digits, the selected GMM was a 9-component mixture with ellipsoidal equal shape covariance matrices (VEV). A 7-component GMM was selected for the group of 6 digits, whereas a 5-component mixture for the group of 9 digits, both with ellipsoidal equal volume and shape covariance matrices (EEV). 
Figure~\ref{fig1:pendigits} shows a static view of a 3D plot of observations projected along the first three MSIR directions (for a rotating 3D plot, see the Supplementary material). The three groups of digits appear to be well separated by both location and variation, with a small separate sub-group of points for digits 9, and some outliers.
Comparing this plot with Figure~3 of \citet{Li:Wang:2007}, we note that the main characteristics of the data are retained, but some other features are also visible, such as the more compact shape for the main group of 9's, and the elongated, curved cluster of 0's.

\begin{figure}[htb]
\centering
\def\baselinestretch{1.1}
\includegraphics[width=0.5\linewidth]{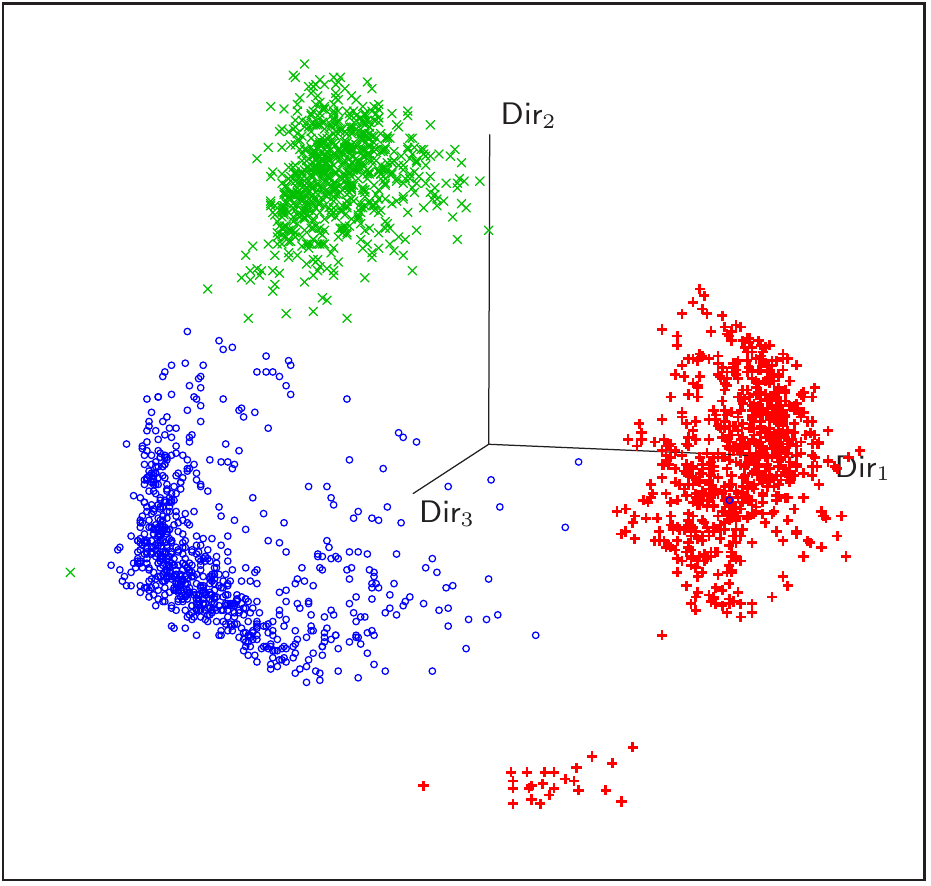}
\caption{Static view of a 3D plot of handwritten digits projectd along first three MSIR directions, with points marked according to digit: \textcolor{blue}{$\circ$} = 0, \textcolor{green}{$\times$} = 6, \textcolor{red}{+} = 9.}
\label{fig1:pendigits}
\end{figure}

One advantage of the MSIR approach is that it allows straightforward classification of observations on the basis of the estimated finite mixtures for each class. 
In the present case, the estimated MSIR model postulates that digits from class $h=\{0,6,9\}$ can be described as 
$\hat{f}(\X|Y=h) = \sum_{k=1}^{K_h} \hat{\pi}_{hk}\phi(\X; \hat{\mub}_{hk}, \hat{\Sigmab}_{hk})$,
with number of components $K_h = \{9,7,5\}$ and covariance matrices $\hat{\Sigmab}_{hk}$ which are parametrized according to models VEV, EEV and EEV, as described in \citet{Fraley:Raftery:2006Mclust}. 
Thus, we may estimate the probability of obtaining a digit $h = \{0,6,9\}$, given predictors $\X$ as follows:
\begin{equation*}
\hat{\Pr}(Y=h|\X) = \frac{\hat{f}(\X|Y=h) \hat{\tau}_h}{\displaystyle\sum_{l = \{0,6,9\}} \hat{f}(\X|Y=l)\hat{\tau}_l},
\end{equation*}
where $\hat{\tau}_l$ are the observed fractions of digits $l$ in the sample.
Recalling that the CDRS is subspace $\Space({\B})$ so that $Y \ind \X|\B\T\X$, the above expression can be expressed equivalently as:
\begin{equation*}
\hat{\Pr}(Y=h|\Zest) = \frac{\hat{f}(\Zest|Y=h) \hat{\tau}_1}{\displaystyle\sum_{l= \{0,6,9\}} \hat{f}(\Zest|Y=l)\hat{\tau}_l},
\end{equation*}
where $\Zest = \Xobs\hat{\B}$ are the MSIR variates and
$\hat{f}(\Zest|Y=h) = \sum_{k=1}^{K_h} \hat{\pi}_{hk}\phi(\Zest;  \hat{\B}\T\hat{\mub}_{hk},  \hat{\B}\T\hat{\Sigmab}_{hk} \hat{\B})$. 
Observations, from either the training or test sets, can be classified according to the MAP principle. By Proposition 1, a classification rule can only be based on a subset of the most important directions.
Figure~\ref{fig2:pendigits} shows the error rates for classifying digits from the training and test sets as a function of CDRS dimension. The smallest error rate is achieved when $d=3$, i.e., when the first three MSIR directions are used.

Table~\ref{tab1:pendigits} shows the training and test error rates for some classification methods: (i) classical linear discriminant analysis, (ii) discriminant analysis based on Gaussian finite mixture modeling \citep{Fraley:Raftery:2002}, (iii) SIR, obtained by fixing $G=1$ for all classes and using the two estimable directions, and (iv) MSIR using the first three directions. 
The training errors are the same for the first three methods, but the test errors are different, as GMMDA achieves the smallest value. 
Classification based on MSIR provides a larger error rate on the training set, but the smallest classification error on the test set. Thus, in this case, the classification rule based on the MSIR directions appears to be more robust, as it avoids overfitting the training set and achieves a good accuracy on the test set.

\begin{minipage}[t]{\textwidth}
\begin{minipage}[t]{0.55\linewidth}
\def\baselinestretch{1.1}
\includegraphics[width=\textwidth]{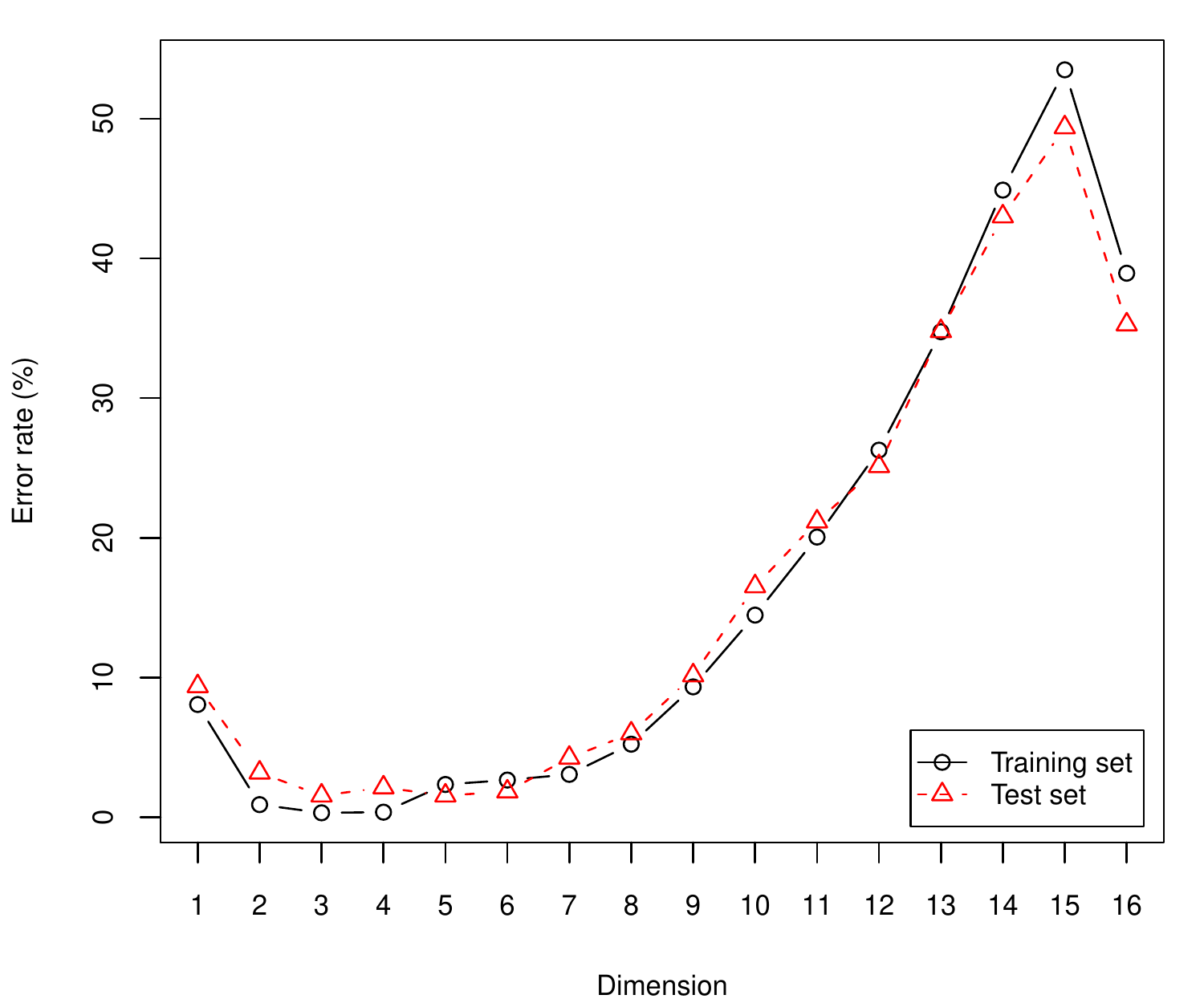}
\figcaption{Classification error rates of MSIR for pen digit data set as a function of dimensionality.}
\label{fig2:pendigits}
\end{minipage}
\qquad
\begin{minipage}[t]{0.35\linewidth}
\vspace*{-19em}
\centering
\tabcaption{Classification error rates for some classifiers based on training and test sets.}
\label{tab1:pendigits}
\smallskip
\begin{tabular}{lcc}
\toprule
 & \multicolumn{2}{c}{Error rate \%} \\
Classifier &  Train &  Test \\
\midrule
LDA           &  0.18 &  2.32 \\
GMMDA         &  0.18 &  2.03 \\
SIR ($d=2$)   &  0.18 &  2.13 \\
MSIR ($d=3$)  &  0.32 &  1.55 \\
\bottomrule
\end{tabular}
\end{minipage}
\end{minipage}

\section{Concluding remarks}

In this paper we propose a model-based approach to dimension reduction which yields a more flexible version of SIR. This is achieved by modeling the distribution within each slice through a finite mixture of Gaussian densities. 
The algorithm for MSIR estimation, determination of dimensionality, and some other results are presented. 
The favorable behavior of MSIR with respect to other popular dimension reduction methods are shown through extensive simulation studies. 
In particular, MSIR overcomes the main limitation of standard SIR in dealing with symmetric relationships. 
Compared with SAVE, MSIR is more efficient and has higher accuracy in the case of linear trends. 
Its performance, particularly for correlated predictors, is also competitive with, or superior to, that of DR, which is reported by \citet{Li:Wang:2007} as the most accurate dimension reduction method based on the first two inverse moments.

\cite{Cook:Forzani:2009} recently introduced a likelihood-based dimension reduction method under the assumption of conditional normality of predictors given the response. Numerical optimization was used for maximization of the log-likelihood on Grassman manifolds. There are similarities between the two methods, but also some substantial differences. In particular, their proposal assumes $\X|Y \sim N(\mub_y, \Deltab_y)$, where both mean and covariance matrix depend on the response variable. Different structures for $\mub_y$ and $\Deltab_y$ yield different models. In MSIR, we employed the flexibility of finite mixture of Gaussian densities to approximate the distribution of $\X|Y$, with data-driven selection of the number of components and the covariance structure. 
Another recent proposal by \cite{Wang:Yin:2011} introduces the use of orthogonal series to estimate the inverse mean space. 
The relative merits and a thorough comparison of these approaches compared with our proposal is an area for further research.

In this paper, we deal with the standard setting, in which the number of observations is larger than the number of predictors.
However, in the case of $p \gg n$, we need to account for possible singularities in the estimation of covariance matrices, arising both from the fitting of Gaussian mixture models and the marginal distribution of the predictors. 
This can be done by imposing restrictions on the possible form of covariance structures, i.e., assuming spherical or diagonal covariance matrices.

Finally, we point out that there are some open issues which deserve further study, as, for instance, the sensitivity of MSIR to the violation of the linearity condition, the applicability in case of high-dimensional predictors, the investigation of other criteria for selecting the mixture model parametrization and number of components within slice.

Supplementary materials including further tables and graphs of simulation results are available from the author's web page.
An \textsf{R} package called \texttt{msir} implementing the method proposed in this paper is available on the Comprehensive R Archive Network at \url{http://CRAN.R-project.org/package=msir}.

\bibliographystyle{spbasic}
\bibliography{msir}

\end{document}